\documentclass[12pt]{article}
\usepackage[paperwidth=210mm, paperheight=297mm,bindingoffset=0mm]{geometry}
\usepackage[T1]{fontenc}
\usepackage{setspace}
\usepackage{tgtermes}
\usepackage{bbm}
\usepackage{mathbbol}
\usepackage[]{enumitem}
\setlist[enumerate,1]{itemsep=0pt, parsep=0pt, listparindent=\parindent}
\setlist[enumerate,2]{ref=\theenumii, itemsep=0pt, parsep=0pt, listparindent=\parindent}
\setlist[itemize,1]{itemsep=0pt, parsep=0pt, listparindent=\parindent}
\setlist[itemize,2]{itemsep=0pt, parsep=0pt, listparindent=\parindent}
\usepackage{titling}
\usepackage{hyperref}
\usepackage{mathtools}
\usepackage[backend=biber,
style=numeric,
sorting=ynt
]{biblatex}
\usepackage[nottoc,numbib]{tocbibind}
\usepackage{listings}
\usepackage{diagbox}
\usepackage{pifont}

\usepackage{xcolor}
\usepackage[bottom]{footmisc} 
\usepackage{amssymb}
\usepackage{amsthm}
\usepackage{blindtext}
\usepackage{etoolbox}
\usepackage{graphicx}              
\usepackage{amsmath}               
\usepackage{amsfonts}              
\usepackage{amsthm}                
\usepackage{hyperref}
\hypersetup{citecolor=red}
\usepackage{cleveref}
\numberwithin{equation}{section}
\usepackage{caption}
\usepackage{float}

\hypersetup{
colorlinks=true,
linkcolor=blue
}

\title{Relative Constructibility via Generalised Sequential Algorithms}
\author{Desmond Lau}

\begin{document}

\maketitle

\begin{abstract}
We modify Gurevich's definition of sequential algorithms, so that it becomes amenable to computation with arbitrarily large sets on a sufficiently intuitive level. As a result, two classes of abstract algorithms are obtained, namely generalised sequential algorithms (GSeqAs) and generalised sequential algorithms with parameters (GSeqAPs). We derive from each class a relative computability relation on sets which is analogous to the Turing reducibility relation on reals. We then prove that the relative computability relation derived from GSeqAPs is equivalent to the relative constructibility relation in set theory.
\end{abstract}

\newtheorem{thm}{Theorem}[section]
\newtheorem{innercustomlem}{Lemma}
\newenvironment{customlem}[1]
  {\renewcommand\theinnercustomlem{#1}\innercustomlem}
  {\endinnercustomlem}
\newtheorem{innercustomdef}{Definition}
\newenvironment{customdef}[1]
  {\renewcommand\theinnercustomdef{#1}\innercustomdef}
  {\endinnercustomdef}
\newtheorem{lem}[thm]{Lemma}
\newtheorem{prop}[thm]{Proposition}
\newtheorem{cor}[thm]{Corollary}
\newtheorem{conj}[thm]{Conjecture}
\newtheorem{ques}[thm]{Question}
\newtheorem*{claim}{Claim}
\theoremstyle{definition}
\newtheorem{defi}[thm]{Definition}
\theoremstyle{remark}
\newtheorem*{rem*}{Remark}
\newtheorem{rem}[thm]{Remark}
\newtheorem{ex}[thm]{Example}
\newtheorem{ob}[thm]{Observation}
\newtheorem{fact}[thm]{Fact}
\newtheorem{con}[thm]{Convention}
\newtheorem{diff}[thm]{Difficulty}

\newcommand{\bd}[1]{\mathbf{#1}}  
\newcommand{\RR}{\mathbb{R}}      
\newcommand{\ZZ}{\mathbb{Z}}      
\newcommand{\col}[1]{\left[\begin{matrix} #1 \end{matrix} \right]}
\newcommand{\comb}[2]{\binom{#1^2 + #2^2}{#1+#2}}
\newcommand{\eq}{=}

\newcommand{\blankpage}{
\newpage
\thispagestyle{empty}
\mbox{}
\newpage
}

{\let\clearpage\relax \tableofcontents} 
\thispagestyle{empty}

\section{Introduction}\label{sect1}

Models of computation studied by computer scientists and mathematicians are aplenty. They include common programming languages, along with the semantics (see e.g.\cite{scott} and \cite{meyer}) --- in various paradigms --- associated with each language. They also include stalwarts from almost a century ago, three of which are
\begin{enumerate}[leftmargin=40pt, label=(CT\arabic*)]
    \item\label{ct1} G\"{o}del's operation schema,
    \item Church's $\lambda$-calculus, and of course
    \item\label{ct3} Turing machines.
\end{enumerate}
That all these models of computation possess the same computational power is great evidence for the Church-Turing thesis. However, modern intuition of computation has evolved to become rather high-level, and the formal models \ref{ct1} to \ref{ct3} which motivated the Church-Turing thesis have been relegated to play the roles of convenient mathematical foundations, instead of anything resembling practical mental models of computation.

Perhaps the most celebrated attempt at amalgamating and formalising useful mental models of computations was born out of this realisation. In a series of papers (starting with \cite{gurevich}), Gurevich gave (closely related) semantic axiomatisations of abstract algorithms, and showed how these algorithms are captured by sufficiently machine-like models called \textit{abstract state machines}. Abstract state machines were defined by Gurevich himself way back in the 1980s, although more organised literature thereof only became available in the 1990s. A full definition of these machines can be found in e.g. \cite{evolvingalg}.

Going by the semantic abstraction of Gurevich, algorithms far more closely resemble ``high-level'' meta-theoretic representations of Plotkin-style semantics (found e.g. in \cite{plotkin}), than they do anything among \ref{ct1} to \ref{ct3}. That their syntactic counterparts --- the programming languages associated with abstract state machines --- have been adapted for use in fields ranging from software development to systems engineering via (what is now known as) the \textit{ASM method} is testament to the intuitiveness of Gurevich's design. Historical details regarding the ASM method can be found in \cite{borger} and Chapter 9 of \cite{borgerbook}.

One obvious commonality of all models of computation hitherto mentioned is that they admit only finitary computation, which is understandable and justified, considering infinitary computations are unrealisable in the real world. But are there inherent technical difficulties in formalising infinitary computation? The answer is ``No, not really,'' as evidenced by a multitude of well-received attempts (as much as a subject this niched can be well-received) at generalising computation to the infinite. Such generalisations are appropriately termed \textit{models of generalised computation}.

Since many decades ago, mathematical logicians have tried their hand at adapting models among \ref{ct1} to \ref{ct3} to allow computations with and on arbitrary sets.

From the perspective of recursion theory, where the structure of Turing degrees under Turing reducibility and the jump operator is the main object study, leading candidates include $\alpha$-recursion and $E$-recursion. Both paradigms depend on models of computation that generalises \ref{ct1} with schemata on arbitrary sets (see e.g. \cite{takeuti} and \cite{normann}). These schemata give rise to analogues of important results in classical recursion theory, on sets that can be much larger than reals. 

From the perspective of ``computer as a machine'', we have Koepke's ordinal Turing machines (see \cite{koepke1}). Ordinal Turing machines are a generalisation of \ref{ct3}: they work like standard Turing machines, except their tapes are now $ORD$-length and they can each refer to finitely many ordinals as parameters. 

In any case, it appears that the models of generalised computation introduced so far rely very much on the Church-Turing thesis for justification of properness. Such justification usually goes along the lines of,
\begin{quote}
    \emph{since $X$ nicely generalises a model among \ref{ct1} to \ref{ct3} and the subject of generalisation is a model of computation by the Church-Turing thesis, $X$ should also be a model of (generalised) computation.}
\end{quote}
This kind of arguments often lead to unintuitive definitions insofar as computation is concerned, sometimes to the extent of an apology being issued (e.g. \cite{sackserec}). In fact, it can be argued that the said models were formulated with specific properties of relative computability --- instead of intuitiveness --- in mind, to ensure their degrees of computability are equipped with non-trivial structure. 

In this paper, we adapt Gurevich's \textit{sequential algorithms} for generalised computation. In doing so, we bypass the appeal to the Church-Turing thesis, and provide a ``top-down'' counterpart to the typical ``bottom-up'' approach to modelling generalised computation. We make gradual well-meaning modifications to Gurevich's definition, so that even as our algorithms become more encompassing, the growth in their computational power does not contradict our intuition of computation born out of classical computing. In Subsection \ref{ss210}, we see that plainly omitting the \textit{bounded exploration postulate} from the initial definition of sequential algorithms is wholly unsatisfactory from a set-theoretic standpoint. The remainder of Section \ref{sect2} is thus geared towards finding the right properties to endow a generalised sequential algorithm with, properties that can plug the gap left behind by the aforementioned omission. Notable milestones in this search/journey include:
\begin{itemize}
    \item a very much local, yet infinitary, version of the bounded exploration postulate formulated in Subsection \ref{ss240},
    \item a streamlined method of writing the input and reading the output of a terminating run devised in Subsection \ref{ss250},
    \item the definition of transfinite runs in Subsection \ref{ss270}, accompanied by an argument about the nigh-universality of our method of taking limits, and
    \item the addition of ordinal parameters in Subsection \ref{ss320} for a better behaved relative computability relation.
\end{itemize}

The result of our elaborate adaptation process is a new model of generalised computation called \textit{generalised sequential algorithms with parameters} (abbreviated as \textit{GSeqAPs}). The first draw of GSeqAPs is that their more structured templates facilitate the design of generalised algorithms, without compromising too much of the intuitiveness in Gurevich's conception. Coupled with an input/output paradigm, GSeqAPs naturally induce a relative computability relation, $\leq^P$, on sets of ordinals. It turns out that $\leq^P$ possesses a particular suite of desirable traits not found in other relative computability relations associated with generalised computation, especially under computational constraints (Table \ref{table3}, as well as the following subsection). This is the second draw of GSeqAPs. 

Towards a generalised analogue of the Church-Turing thesis, we also show that $\leq^P$ coincides with the \textit{relative constructibility} relation (Theorem \ref{thm275}). 

\subsection{Generalisations of (Relative) Computability}\label{subsec22}

A notion of computability is meant to chart the limits of what can and cannot be computed. In the classical case, we hold the Church-Turing thesis responsible for this purpose, a responsibility rooted in its convincing portrayal of computer-hood. Indeed, the hallowed status afforded to the Church-Turing thesis stems from the unlikely convergence in power of what seems to be independently-devised models of computation.

Certain models of computation (perhaps with simple modifications) naturally induce a notion of relative computability on some distinguished set. Such a relative computability relation typically extends the notion of computability, in that being computable is equivalent to being minimally relative computable. A most notable example is the Turing reducibility relation induced by oracle machines, which are essentially Turing machines equipped with quick access to a real number serving as an ``oracle''. 

It can be argued that a robust model of generalised computation should induce a (binary) relative computability relation on arbitrary sets, and that the relative computability relation thus induced should satisfy the following three desiderata.
\begin{enumerate}[leftmargin=40pt, label=(D\arabic*)]
    \item\label{d1} \textit{Universality}: the relation should be reflexive on all sets of ordinals.
    \item\label{d2} \textit{Transitivity}: the relation should be transitive.
    \item\label{d3} \textit{Coherence}: if the relative computability of two sets is witnessed under reasonable computational constraints, then the same should be witnessed after lifting these constraints.
\end{enumerate}
Unsurprisingly, each of the well-known models of generalised computation introduced in the previous subsection induces relative computability relation on sets. We shall examine whether these relations satisfy properties \ref{d1} to \ref{d3}. The next three paragraphs contain an elaboration of our examination, and Tables \ref{table1} and \ref{table2}, a summary of its outcomes.

The relative computability relation associated with $\alpha$-recursion, written $\leq_{\alpha}$, is only defined on $\mathcal{P}(\alpha)$, the powerset of $\alpha$. As such, it is not universal. Naively, we can define a universal relation $\leq_A$ by referring to $\leq_{\alpha}$ as $\alpha$ ranges over all \textit{admissible} ordinals: $$A \leq_A B \text{ iff there is } \alpha \text{ for which } A \leq_{\alpha} B \text{.}$$ However, if we interpret $\alpha$ to be a measure of the amount of computational resources available, as is natural, Theodore Slaman's answer in \cite{slaman} tells us $\leq_A$ is doomed to be incoherent. It seems, then, that there is no easy fix to the non-universality of the $\leq_{\alpha}$'s.

On the other hand, the relative computability relation of $E$-recursion, written $\leq_E$, is both universal and transitive out of the box. Further, it is common in $E$-recursion to study ``runs'' restricted to arbitrary $E$\textit{-recursively closed} initial segment of $L$. Restrictions of this type are computational constraints imposed on the evaluation of $\leq_E$. Unfortunately, with respect to these constraints, $\leq_E$ is not known to behave coherently (see e.g. ``\textbf{The Divergence-Admissibility Split}'' on p.11 of \cite{sackserec}).

Taking a leaf from how classical Turing machines were modified to become classical oracle machines, we can straightforwardly augment ordinal Turing machines with set oracles. This allows one to define a relative computability relation $\preceq_A$ on arbitrary sets in the same way real oracles are used to define the Turing reducibility relation on the reals. In fact, with reference to $\preceq_{\alpha}$ in Definition 19 of \cite{koepke2}, $$A \preceq_A B \text{ iff there is } \alpha \text{ for which } A \preceq_{\alpha} B \text{.}$$ Such a definition is easily seen to be universal and coherent. It is stated in Theorem 20 of \cite{koepke2} that restricting this relative computability relation to subsets of admissible ordinals $\alpha$ may result in a lack of transitivity. We can however verify that $\preceq_{\alpha}$ is transitive whenever $\alpha$ is regular, and in turn, use it to conclude the transitivity of $\preceq_A$.

\begin{table}[!ht]
    \caption[Relative computable relations across three major models of generalised computation]{The relative computable relations native to three major models of generalised computation.}
    \label{table1}
    \centering
    \begin{tabular}{|l||*{3}{c|}}\hline
        \backslashbox[90pt]{\footnotesize Property}{\footnotesize Relation}
        &\makebox[4em]{$\leq_{\alpha}$} &\makebox[2em]{$\leq_E$} &\makebox[4em]{$\preceq_{\alpha}$} \\\hline\hline
        Universal? & \ding{55} & \ding{51} & \ding{55} \\\hline
        Transitive? & \ding{51} & \ding{51} & \ding{55} \\\hline
        Coherent? & \ding{55} & \ding{55} & \ding{51} \\\hline
        Appears in$\dots$ & $\alpha$-recursion & $E$-recursion & $\alpha$-computability \\\hline
    \end{tabular}
\end{table}

\begin{table}[!ht]
    \caption[Comparison of two universal relative computability relations derived from non-universal relations]{Comparison of two universal relative computability relations derived from separate classes of non-universal relations.}
    \label{table2}
    \centering
    \begin{tabular}{|l||*{2}{c|}}\hline
        \backslashbox[90pt]{\footnotesize Property}{\footnotesize Relation}
        &\makebox[4em]{$\leq_A$} &\makebox[4em]{$\preceq_A$} \\\hline\hline
        Universal? & \ding{51} & \ding{51} \\\hline
        Transitive? & \ding{51} & \ding{51} \\\hline
        Coherent? & \ding{55} & \ding{51} \\\hline
        Derived from$\dots$ & the $\leq_{\alpha}$'s & the $\preceq_{\alpha}$'s \\\hline
    \end{tabular}
\end{table}

Like $\preceq_A$, the relative computability relation associated with GSeqAPs, $\leq^P$, also fulfils desiderata \ref{d1} to \ref{d3}. What makes $\leq^P$ stand out is that unlike $\preceq_A$, it tends very much to remain transitive under computational constraints. More precisely, for all admissible ordinals $\alpha$, both
\begin{itemize}
    \item $\leq^P_{\alpha}$, the restriction of $\leq^P$ to computations requiring $\alpha$-sized space, and
    \item $\leq^{P, s}_{\alpha}$, the restriction of $\leq^P$ to computations requiring $\alpha$-sized space and terminating in $< \alpha$-many time steps,
\end{itemize}
are at once, transitive and coherent (Propositions \ref{prop251}, \ref{prop252}, \ref{prop270n} and \ref{prop271n}). It is interesting to note that our formulation of GSeqAPs, in accordance to Gurevich's principles of inclusiveness and intuitiveness espoused in \cite{gurevich}, inadvertently gives rise to a series of novel, non-trivial and well-behaved relative computability relations. 

\section{Adapting Sequential Algorithms for the Infinite}\label{sect2}

In this section, we progressively modify the definition of a sequential algorithm, to arrive at a more concrete notion of algorithms which also better captures our intuition of generalised computation. 

\subsection{Formal Sequential Algorithms}\label{ss210}

Gurevich introduced the \textit{ASM thesis} in \cite{gurevich} as a justification of correctness for his definition of (a most general conceptualisation of) a sequential algorithm. He followed it up with a number of collaborations and papers (e.g. \cite{gurevichblass}) on various other classes of algorithms. If we follow the convention of using ``algorithm'' and ``computer program'' interchangeably, then sequential algorithms are natural candidates for capturing a general notion of computability. However, it seems that Gurevich himself did not intend for sequential algorithms to capture computation on an infinitary scale. We shall start this section with a brief examination on how his definition might be ill-equipped at dealing with generalised computation. 

Briefly, a \emph{sequential algorithm} --- henceforth \textit{SeqA}, to distinguish the formal concept from the colloquial --- $\mathfrak{M}$ comprises a class of states $S(\mathfrak{M})$, a class of initial states $I(\mathfrak{M})$, a class transition function $\tau_{\mathfrak{M}} : S(\mathfrak{M}) \longrightarrow S(\mathfrak{M})$, and a finite set $G$ of ground terms over $\sigma(\mathfrak{M})$ such that, letting $G^s$ denote the function $\dot{x} \mapsto \dot{x}^s$ with domain $G$,
\begin{enumerate}[label=(A\arabic*)]
    \item\label{a1} every $s \in S(\mathfrak{M})$ is a first-order structure with the same finite signature, $\sigma(\mathfrak{M})$,
    \item\label{a2} $I(\mathfrak{M}) \subset S(\mathfrak{M})$,
    \item\label{a3'} for all $s \in S(\mathfrak{M})$, $\tau_{\mathfrak{M}}(s)$ has the same base set $b_{\mathfrak{M}}(s)$ as $s$
    \item\label{a3} for all $s \in S(\mathfrak{M})$, $\ulcorner X \urcorner \in \sigma(\mathfrak{M})$, $n < \omega$ and $\Vec{a} \in b_{\mathfrak{M}}(s)^n$, if $n$ is the arity of $\ulcorner X \urcorner$, then
    \begin{equation*}
        \ulcorner X \urcorner^s(\Vec{a}) \neq \ulcorner X \urcorner^{\tau_{\mathfrak{M}}(s)}(\Vec{a}) \implies \Vec{a}^\frown(\ulcorner X \urcorner^{\tau_{\mathfrak{M}}(s)}(\Vec{a})) \in ran(G^s)^{n+1}
    \end{equation*}
    whenever $\ulcorner X \urcorner$ is a function or constant symbol and
    \begin{equation*}
        \neg (\ulcorner X \urcorner^s(\Vec{a}) \iff \ulcorner X \urcorner^{\tau_{\mathfrak{M}}(s)}(\Vec{a})) \implies \Vec{a} \in ran(G^s)^n
    \end{equation*}
    whenever $\ulcorner X \urcorner$ is a relation symbol, and
    \item\label{a4} for all $s, s' \in S(\mathfrak{M})$, $\ulcorner X \urcorner \in \sigma(\mathfrak{M})$ and $n < \omega$, if $n$ is the arity of $\ulcorner X \urcorner$ and $G^s = G^{s'}$, then whenever $\Vec{a} \in ran(G^s)^n$,
    \begin{align*}
        \ulcorner X \urcorner^{s} \text{ agrees with } \ulcorner X \urcorner^{\tau_{\mathfrak{M}}(s)} \text{ at } \Vec{a} \iff \ & \ulcorner X \urcorner^{s'} \text{ agrees with } \ulcorner X \urcorner^{\tau_{\mathfrak{M}}(s')} \text{ at } \Vec{a} \text{, and} \\
        \ulcorner X \urcorner^{s} \text{ disagrees with } \ulcorner X \urcorner^{\tau_{\mathfrak{M}}(s)} \text{ at } \Vec{a} \implies \ & \ulcorner X \urcorner^{\tau_{\mathfrak{M}}(s)} \text{ agrees with } \ulcorner X \urcorner^{\tau_{\mathfrak{M}}(s')} \text{ at } \Vec{a} \text{.}
    \end{align*}
\end{enumerate}
A (finite) terminating run of $\mathfrak{M}$ is a sequence of states $(s_i : i \leq n < \omega)$ for which
\begin{enumerate}[label=(B\arabic*)]
    \item $s_0 \in I(\mathfrak{M})$,
    \item $\tau_{\mathfrak{M}}(s_n) = s_n$, and
    \item for all $m < n$, $s_{m+1} = \tau_{\mathfrak{M}}(s_m) \neq s_m$.
\end{enumerate}

For brevity's sake, some details are withheld here. First, Gurevich originally require states to have only function and constant symbols in their signatures. However, this restriction is purely cosmetic because every relation in a structure can be represented as a boolean function. Second, $S(\mathfrak{M})$ is supposed to be closed under all possible isomorphisms, thus becoming a proper class. This closure property allows one to derive \ref{a3} and \ref{a4} from Gurevich's \emph{bounded exploration postulate} (see \cite{gurevich}, Lemma 6.2 for a derivation). Even though our presentation might seem logically weaker than Gurevich's original, the two are actually provably equivalent, modulo closure under isomorphisms. In fact, going through \ref{a3} and \ref{a4} is standard for proofs equating the computational power of different paradigms through sequential algorithms (in e.g. \cite{dershowitz} or even \cite{gurevich} itself), so our sidestepping of set-theoretic concerns regarding the sizes of state spaces is accompanied by some conveniences. Finally, our definition above can be generalised so that $\tau_{\mathfrak{M}}$ need only be a binary relation -- resulting in algorithms which might not be sequential -- but that is unnecessary if we only want to capture enough of generalised computation for a sufficiently general relative computability relation.

Intuitively, a state of an algorithm is a configuration of some generalised Turing machine. Each initial state corresponds to an initial configuration, from which one can read off the input (or one may take the initial state itself to be the input). The last state of a terminating run, sometimes called a \emph{final state}, corresponds in spirit to an end configuration of a Turing machine run, from which one can read off the output. For example, if we have a specific symbol $\ulcorner X \urcorner$ in the signature to represent an analogue of the Turing machine tape, then a terminating run $(s_1, \dots, s_n)$ should give an algorithmic transformation of the input representation $X^{s_1}$ into the output representation $X^{s_n}$.

Now, requirements \ref{a3} and \ref{a4} set an \textit{a priori} finite bound to the number of bits that an algorithm can alter in each step of its computational run. Since a terminating run is presumed to be finite, there is no way for an algorithm to make an infinite amount of change before terminating. It is therefore very tempting to just free ourselves from the shackles of \ref{a3} and \ref{a4}, especially since we aim to have as general a definition of an generalised algorithm as possible. Unfortunately, this turns out to be \textit{too} inclusive. Consider the example below, modelled after a Turing machine with a $\kappa$-length tape.

\begin{ex}\label{ex21}
Let $\kappa$ be an infinite cardinal. Expand the signature of $(\kappa; \in)$ to include a unary relation symbol $\ulcorner A \urcorner$. Now if $A_1$ and $A_2$ are subsets of $\kappa$, then we can define $\mathfrak{M}(\kappa, A_1, A_2)$ as follows:
\begin{gather*}
    s_1 := (\kappa; \in, A_1) \\
    s_2 := (\kappa; \in, A_2) \\
    S(\mathfrak{M}(\kappa, A_1, A_2)) := \{s_1, s_2\} \\
    I(\mathfrak{M}(\kappa, A_1, A_2)) := \{s_1\} \\
    \tau_{\mathfrak{M}(\kappa, A_1, A_2)} := \{(s_1, s_2), (s_2, s_2)\} \text{,}
\end{gather*}
so that $A_1$ and $A_2$ both interpret $\ulcorner A \urcorner$ on the base set $\kappa$ shared between the states $s_1$ and $s_2$. 
\end{ex}

Fix subsets $A_1$ and $A_2$ of $\kappa$. Note that $\mathfrak{M}(\kappa, A_1, A_2)$ satisfies requirements \ref{a1} to \ref{a3'}. Moreover, the only terminating run of $\mathfrak{M}(\kappa, A_1, A_2)$ is $(s_1, s_2)$. As the interpretation of the symbol $\ulcorner A \urcorner$ is the only difference between $s_1$ and $s_2$, the moral implication of this run is that we can algorithmically get from $A_1$ to $A_2$. If we set $A_1$ to be a simply definable set like the empty set, then every subset of $\kappa$ is algorithmically derivable from $A_1$. This seems incredible given the information content of arbitrary sets of ordinals, as exemplified by the next proposition.

\begin{prop}\label{prop22}
Let $M$ be a transitive model of $\mathsf{ZFC}$ and $X \in M$. Then there is a set of ordinals $c \in M$ such that if $N$ is any transitive model of $\mathsf{ZF}$ containing $c$, then $X \in N$. 
\end{prop}

\begin{proof}
Let $Y'$ be the transitive closure of $X$ (under the membership relation $\in$) and set $Y := Y' \cup \{X\}$. Then $Y$ is $\in$-transitive. Choose a bijection $f$ from a cardinal $\kappa$ into $Y$. Use $\in'$ to denote the unique binary relation $R$ on $\kappa$ such that
\begin{equation*}
    R(\alpha, \beta) \iff f(\alpha) \in f(\beta) \text{.}
\end{equation*}
Now apply G\"{o}del's pairing function to code $\in'$ as a subset $c$ of $\kappa$. To recover $X$ from $c$, first apply the inverse of the pairing function followed by the Mostowski collapse to get $Y$. Then $X$ is definable from $Y$ as the unique $\in$-maximal element of $Y$. This decoding process is absolute for transitive models of $\mathsf{ZF}$ because all its components are.
\end{proof}

The procedure described in the proof of Proposition \ref{prop22} allows us to encode any set as a set of ordinals, such that the decoding can be done in an absolute manner. If we believe that such a decoding scheme (particularly, G\"{o}del's pairing function) is algorithmic, then our aforementioned loosening of the definition of SeqAs entails that we can algorithmically derive every set in $V$ from just $\emptyset$. Since most set theorists are of the opinion that sets in $V$ can be extremely complicated and far from constructible in a very strong sense, this conclusion is wholly unsatisfactory, if not patently false. In other words, just eschewing requirements \ref{a3} and \ref{a4} yields \emph{too} encompassing a paradigm for algorithms, and thus should not be used to draw the line for generalised computability.

\begin{rem}
In Example \ref{ex21}, we can set $A_1$ and $A_2$ to be subsets of $H(\kappa)$, and alter the definition of $\mathfrak{M}(\kappa, A_1, A_2)$ such that 
\begin{gather*}
    s_1 := (H(\kappa); \in, A_1) \\
    s_2 := (H(\kappa); \in, A_2) \text{.}
\end{gather*}
In this way, we obtain a more literal argument that every set is algorithmically derivable from every other set, sidestepping the need for Proposition \ref{prop22} as motivation. Nonetheless, Proposition \ref{prop22} will be invoked in later parts of this section, so it is as good a time now as any to introduce it.
\end{rem}

From now on, and until further notice, let us relax the definition of SeqAs such that they need only satisfy requirements \ref{a1} to \ref{a3'}. For much of the rest of this section, we will make educated modifications and impose suitable restrictions to this relaxed definition of a SeqA, in an attempt to formulate a more philosophically justified notion of generalised algorithms, at least from a set-theoretic viewpoint.

\subsection{Modifying the Satisfaction Relation}

An obvious cause of the triviality identified in Example \ref{ex21}, is the lack of restriction on the transition function $\tau_{\mathfrak{M}}$ of a SeqA $\mathfrak{M}$. Ideally, $\tau_{\mathfrak{M}}(s)$ should be describable is a way which only depends locally on itself and $s$. One formalisation of a ``describable local property'' applicable to first-order structures is the notion of a \emph{first-order property}. A first-order formula $\phi$ over the signature of a structure $\mathfrak{A}$ is a first-order property of $\mathfrak{A}$ iff $\mathfrak{A} \models \phi$, where $\models$ is Tarski's satisfaction relation. To describe a transition function, we want to isolate a first-order property of the next state that depends solely on the current state. However, because a candidate property should refer to two different structures (the current and next states), some slight modifications to the satisfaction relation are in order.

\begin{defi}
Given a signature $\sigma$, define $2 \cdot \sigma$ to be the disjoint union of $\sigma$ and itself. In other words, for each symbol $\ulcorner X \urcorner \in \sigma$ and each $i < 2$, we choose a fresh symbol $\ulcorner X_{i} \urcorner$ and cast it to be of the same type and arity as $\ulcorner X \urcorner$. We then set
\begin{equation*}
    2 \cdot \sigma := \{\ulcorner X_{i} \urcorner : \ulcorner X \urcorner \in \sigma \text{ and } i < 2\} \text{.}
\end{equation*}
\end{defi}

\begin{defi}
A \emph{binary first-order formula} over a signature $\sigma$ is a first-order formula over $2 \cdot \sigma$.
\end{defi}

\begin{defi}
Let $\mathfrak{A}_0$ and $\mathfrak{A}_1$ be first-order structures with the same base set $A$ and signature $\sigma$. For each $a \in A$, choose a fresh constant symbol $\ulcorner a \urcorner$ not in $2 \cdot \sigma$. Denote
\begin{equation*}
    2 \cdot \sigma \cup \{\ulcorner a \urcorner : a \in A\}
\end{equation*}
as $\sigma(A, 2)$.
\end{defi}

\begin{defi}
Let $\mathfrak{A}_0$ and $\mathfrak{A}_1$ be first-order structures with the same base set $A$ and signature $\sigma$. Let $\ulcorner X \urcorner \in \sigma(A, 2)$.
Define
\begin{align*}
    X^{(\mathfrak{A}_0, \mathfrak{A}_1)} := \ & a \text{ if } \ulcorner X \urcorner = \ulcorner a \urcorner \text{ for some } a \in A \text{,} \\
    X^{(\mathfrak{A}_0, \mathfrak{A}_1)} := \ & Y^{\mathfrak{A}_0} \text{ if } \ulcorner X \urcorner = \ulcorner Y_0 \urcorner \text{ for some } \ulcorner Y \urcorner \in \sigma \text{, and} \\
    X^{(\mathfrak{A}_0, \mathfrak{A}_1)} := \ & Y^{\mathfrak{A}_1} \text{ if } \ulcorner X \urcorner = \ulcorner Y_1 \urcorner \text{ for some } \ulcorner Y \urcorner \in \sigma \text{.}
\end{align*}
If $\ulcorner f \urcorner \in \sigma(A, 2)$ is a function symbol and $t$ is a term over $\sigma(A, 2)$ of which interpretation under $(\mathfrak{A}_0, \mathfrak{A}_1)$ is known, then 
\begin{equation*}
    f(t)^{(\mathfrak{A}_0, \mathfrak{A}_1)} := f^{(\mathfrak{A}_0, \mathfrak{A}_1)}(t^{(\mathfrak{A}_0, \mathfrak{A}_1)}) \text{.}
\end{equation*}
\end{defi}

\begin{defi}
Let $\mathfrak{A}_0$ and $\mathfrak{A}_1$ be first-order structures with the same base set $A$ and signature $\sigma$. We define what $(\mathfrak{A}_0, \mathfrak{A}_1) \models_2 \phi$ means for $\phi \in \sigma(A, 2)$ by recursion on the complexity of $\phi$.
\begin{enumerate}[label=(\arabic*)]
    \item If $\phi = \ulcorner t_1 = t_2 \urcorner$ for terms $t_1, t_2$ over $\sigma(A, 2)$, then $$(\mathfrak{A}_0, \mathfrak{A}_1) \models_2 \phi \text{ iff } t_1^{(\mathfrak{A}_0, \mathfrak{A}_1)} = t_2^{(\mathfrak{A}_0, \mathfrak{A}_1)} \text{.}$$
    \item If $\phi = \ulcorner R(t_1, \dots, t_n) \urcorner$ for a relation symbol $\ulcorner R \urcorner \in \sigma(A, 2)$ and terms $t_1, \dots, t_n$ over $\sigma(A, 2)$, then $$(\mathfrak{A}_0, \mathfrak{A}_1) \models_2 \phi \text{ iff } R^{(\mathfrak{A}_0, \mathfrak{A}_1)} (t_1^{(\mathfrak{A}_0, \mathfrak{A}_1)}, \dots, t_n^{(\mathfrak{A}_0, \mathfrak{A}_1)}) \text{.}$$
    \item If $\phi = \ulcorner \neg \psi \urcorner$ for a formula $\psi$ over $\sigma(A, 2)$, then $$(\mathfrak{A}_0, \mathfrak{A}_1) \models_2 \phi \text{ iff } \neg ((\mathfrak{A}_0, \mathfrak{A}_1) \models_2 \psi) \text{.}$$
    \item If $\phi = \ulcorner \psi_1 \wedge \psi_2 \urcorner$ for formulas $\psi_1$ and $\psi_2$ over $\sigma(A, 2)$, then $$(\mathfrak{A}_0, \mathfrak{A}_1) \models_2 \phi \text{ iff } ((\mathfrak{A}_0, \mathfrak{A}_1) \models_2 \psi_1 \wedge (\mathfrak{A}_0, \mathfrak{A}_1) \models_2 \psi_2) \text{.}$$
    \item If $\phi = \ulcorner \exists x \ \psi \urcorner$ for a free variable $\ulcorner x \urcorner$ and a formula $\psi$ over $\sigma(A, 2)$, then $$(\mathfrak{A}_0, \mathfrak{A}_1) \models_2 \phi \text{ iff } \exists a \in A ((\mathfrak{A}_0, \mathfrak{A}_1) \models_2 \psi[x \mapsto a]) \text{,}$$ where $\psi[x \mapsto a]$ is the result of replacing all instances of $\ulcorner x \urcorner$ in $\psi$ by $\ulcorner a \urcorner$.
\end{enumerate}
\end{defi}

Before adding more requirements to the transition function of a SeqA, let us enforce that every state in a SeqA $\mathfrak{M}$ must have the same base set $b(\mathfrak{M})$. This can be done without loss of generality in the definition of a SeqA, because we can extend the base set of each state in $\mathfrak{M}$ to be the union of the base sets of states in $\mathfrak{M}$, while arbitrarily extending the interpretation of function and relation symbols in the signature by each state.

We stipulate that the transition function $\tau_{\mathfrak{M}}$ of a SeqA $\mathfrak{M}$ must be finitely describable over $2 \cdot \sigma(\mathfrak{M})$. That is, there must exist a first-order sentence $\phi^{\tau}_{\mathfrak{M}}$ over $2 \cdot \sigma(\mathfrak{M})$ such that for all $s_1, s_2 \in S(\mathfrak{M})$, 
\begin{equation*}
    \tau_{\mathfrak{M}}(s_1) = s_2 \iff (s_1, s_2) \models_2 \phi^{\tau}_{\mathfrak{M}} \text{.}
\end{equation*}

This new imposition on transition functions seems sufficiently severe to rule out the cases in Example \ref{ex21}. Or does it?

\begin{ex}[\ref{ex21} revisited]
Let $\kappa$, $A_1$, $A_2$, $s_1$, $s_2$ and $\mathfrak{M}(\kappa, A_1, A_2)$ be as in Example \ref{ex21}. Set 
\begin{gather*}
    \phi^{\tau}_{\mathfrak{M}(\kappa, A_1, A_2)} := \ulcorner \forall x \ (x \in_1 c_1 \iff (\forall y \ (\neg (y \in_1 x)))) \urcorner \text{.}
\end{gather*}
Then given any two expansions $s'_1$ and $s'_2$ of $(\kappa; \in)$ with signature $\sigma$ containing $\ulcorner c \urcorner$, 
\begin{equation*}
    (s'_1, s'_2) \models_2 \phi^{\tau}_{\mathfrak{M}(\kappa, A_1, A_2)} \text{ iff } s'_2 \text{ interprets } \ulcorner c \urcorner \text{ to be } 1 \text{.}
\end{equation*}
In particular, $(s_1, s_2) \models_2 \phi^{\tau}_{\mathfrak{M}(\kappa, A_1, A_2)}$ and $(s_2, s_2) \models_2 \phi^{\tau}_{\mathfrak{M}(\kappa, A_1, A_2)}$, so $\tau_{\mathfrak{M}(\kappa, A_1, A_2)}$ is still a valid transition function. 
\end{ex}

It turns out that we can intentionally choose the set of states in a SeqA to trivialise the definition of its transition function. To overcome this, we make it so that the set of states depends only on the base set and signature of the structures involved. What this means is, if $\mathfrak{M}_1$ and $\mathfrak{M}_2$ are SeqAs, $s_1 \in S(\mathfrak{M}_1)$, $s_2 \in S(\mathfrak{M}_2)$ and $s_1$, $s_2$ share the same base set and signature, then $S(\mathfrak{M}_1) = S(\mathfrak{M}_2)$. To wit, we include $b(\mathfrak{M})$ as a primary datum of a SeqA, and derive $S(\mathfrak{M})$ from $b(\mathfrak{M})$ and $\sigma(\mathfrak{M})$ as follows: 
\begin{equation*}
    S(\mathfrak{M}) := \{s : s \text{ is a first-order structure with base set } b(\mathfrak{M}) \text{ and signature } \sigma(\mathfrak{M})\} \text{.}
\end{equation*}

Now, there may be a nagging worry that we over-corrected, that our modified definition of a SeqA is too restrictive. Indeed, this shortcoming becomes obvious when we try to define a SeqA corresponding to a Turing machine.

\begin{diff}\label{diff29}
Consider a Turing machine $M$ with an $\omega$-length tape and $n$ states, for some $n < \omega$. A SeqA $\mathfrak{M}$ representing $M$ naturally has $b(\mathfrak{M})$ include $\omega$ and the set of $M$'s states. A suitable transition function $\tau_{\mathfrak{M}}$ ought to speak of $M$'s states, and to do that it is imperative that each of these states identify uniquely with a term over $\sigma(\mathfrak{M})$. For example, we can have a constant symbol $\ulcorner c_i \urcorner \in \sigma(\mathfrak{M})$ uniquely represent the $i$-th state of $M$. 

However, as members of $S(\mathfrak{M})$ are allowed to freely interpret the $\ulcorner c_i \urcorner$s, there would be a problem of singling out the ``next state'' of $\mathfrak{M}$ based on an arbitrary ``current state'', using the relation $\models_2$ alone.
\end{diff}

Difficulty \ref{diff29} tells us $S(\mathfrak{M})$ should depend not just on $b(\mathfrak{M})$ and $\sigma(\mathfrak{M})$, but also on how symbols in $\sigma(\mathfrak{M})$ are to be interpreted in $b(\mathfrak{M})$. This leads us to the objects of focus in the next subsection.

\subsection{Theories with Constraints in Interpretation}\label{ssect23}

\emph{Theories with constraints in interpretation} (henceforth, \emph{TCI}s) were conceived as a convenient means of looking at generic objects produced by set-theoretic forcing. As TCIs impose natural constraints on first-order theories, they can be easily utilised to constrain the possible states of a SeqA. In fact, TCIs are a generalisation of what Dershowitz and Gurevich call \textit{static functions} in \cite{dershowitz}.

\begin{defi}(Lau, \cite{myself})
A \emph{first-order theory with constraints in interpretation} (\emph{first-order TCI}) --- henceforth, just \emph{theory with constraints in interpretation} (\emph{TCI}) --- is a tuple $(T, \sigma, \dot{\mathcal{U}}, \vartheta)$, where
\begin{enumerate}[label=(\alph*)]
    \item $T$ is a first order theory with signature $\sigma$,
    \item $\dot{\mathcal{U}}$ is a unary relation symbol not in $\sigma$,
    \item $\vartheta$ is a function (the \emph{constraint function}) with domain $\sigma \cup \{\dot{\mathcal{U}}\}$, 
    \item if $x \in ran(\vartheta)$, then there is $y$ such that 
    \begin{enumerate}[label=(\roman*)]
        \item either $x = (y, 0)$ or $x = (y, 1)$, and
        \item if $\vartheta(\dot{\mathcal{U}}) = (z, a)$, then $y \subset z^n$ for some $n < \omega$, and
    \end{enumerate}
    \item if $\vartheta(\dot{\mathcal{U}}) = (z, a)$, then 
    \begin{enumerate}[label=(\roman*)]
        \item $z \cap z^n = \emptyset$ whenever $1 < n < \omega$, and
        \item $z^m \cap z^n = \emptyset$ whenever $1 < m < n < \omega$.
    \end{enumerate}
\end{enumerate}
We call members of the interpretation constraint map \emph{interpretation constraints}.
\end{defi}

\begin{defi}(Lau, \cite{myself})
Let $(T, \sigma, \dot{\mathcal{U}}, \vartheta)$ be a TCI. We say $$\mathcal{M} := (U; \mathcal{I}) \models^* (T, \sigma, \dot{\mathcal{U}}, \vartheta)$$ --- or $\mathcal{M}$ \emph{models} $(T, \sigma, \dot{\mathcal{U}}, \vartheta)$ --- iff all of the following holds:
\begin{enumerate}[label=(\alph*)]
    \item $\mathcal{M}$ is a structure,
    \item $\sigma$ is the signature of $\mathcal{M}$,
    \item $\mathcal{M} \models T$,
    \item if $\vartheta(\dot{\mathcal{U}}) = (y, 0)$, then $U \subset y$,
    \item if $\vartheta(\dot{\mathcal{U}}) = (y, 1)$, then $U = y$, and
    \item for $\dot{X} \in \sigma$,
    \begin{enumerate}[label=(\roman*)]
        \item if $\dot{X}$ is a constant symbol and $\vartheta(\dot{X}) = (y, z)$, then $\mathcal{I}(\dot{X}) \in y \cap U$,
        \item if $\dot{X}$ is a $n$-ary relation symbol and $\vartheta(\dot{X}) = (y, 0)$, then $\mathcal{I}(\dot{X}) \subset y \cap U^{n}$,
        \item if $\dot{X}$ is a $n$-ary relation symbol and $\vartheta(\dot{X}) = (y, 1)$, then $\mathcal{I}(\dot{X}) = y \cap U^{n}$,
        \item if $\dot{X}$ is a $n$-ary function symbol and $\vartheta(\dot{X}) = (y, 0)$, then $$\{z \in U^{n+1} : \mathcal{I}(\dot{X})(z \! \restriction_n) = z(n)\} \subset y \cap U^{n+1}, \text{ and}$$
        \item if $\dot{X}$ is a $n$-ary function symbol and $\vartheta(\dot{X}) = (y, 1)$, then $$\{z \in U^{n+1} : \mathcal{I}(\dot{X})(z \! \restriction_n) = z(n)\} = y \cap U^{n+1}.$$
    \end{enumerate}
\end{enumerate}
We say $(T, \sigma, \dot{\mathcal{U}}, \vartheta)$ has a model if there exists $\mathcal{M}$ for which $\mathcal{M} \models^* (T, \sigma, \dot{\mathcal{U}}, \vartheta)$.
\end{defi}

To specify a SeqA $\mathfrak{M}$, we should also specify a TCI $\mathfrak{T}(\mathfrak{M}) = (T, \sigma, \dot{\mathcal{U}}, \vartheta)$ such that 
\begin{gather*}
    T = \emptyset \\
    \sigma = \sigma(\mathfrak{M}) \\
    \vartheta(\dot{\mathcal{U}}) = (b(\mathfrak{M}), 1).
\end{gather*}
We can then stipulate that 
\begin{equation*}
    S(\mathfrak{M}) := \{s : s \models^* \mathfrak{T}(\mathfrak{M})\} \text{.}
\end{equation*}

Let us now check that any Turing machine can be easily represented as a SeqA.

\begin{ex}\label{ex212}
As in Difficulty \ref{diff29}, let $M$ be a Turing machine with an $\omega$-length tape and $n$ states, for some $0 < n < \omega$. Let the set of $M$'s states be $\{c_i : i < n\}$, with $c_0$ being the only initial state and $c_{n-1}$ the only final state. 

Choose a binary relation symbol $\ulcorner \in \urcorner$, unary relation symbols $\dot{\mathcal{U}}$, $\ulcorner R \urcorner$, unary function symbols $\ulcorner \mathrm{Pre} \urcorner$, $\ulcorner \mathrm{Suc} \urcorner$, and constant symbols $\ulcorner h \urcorner$, $\ulcorner t \urcorner$, all of them distinct from one another. Set 
\begin{gather*}
    b(\mathfrak{M}) := \omega \\
    \sigma(\mathfrak{M}) := \{\ulcorner \in \urcorner, \ulcorner \mathrm{Pre} \urcorner, \ulcorner \mathrm{Suc} \urcorner, \ulcorner h \urcorner, \ulcorner t \urcorner, \ulcorner R \urcorner\} \text{.}
\end{gather*}
Define a function $\vartheta$ on $\sigma(\mathfrak{M}) \cup \{\dot{\mathcal{U}}\}$ as follows:
\begin{gather*}
    \vartheta(\dot{\mathcal{U}}) := (\omega, 1) \\
    \vartheta(\ulcorner \in \urcorner) := (\in \cap \ \omega, 1) \\
    \vartheta(\ulcorner \mathrm{Pre} \urcorner) := (\mathrm{Pre}, 1) \\
    \vartheta(\ulcorner \mathrm{Suc} \urcorner) := (\mathrm{Suc}, 1) \\
    \vartheta(\ulcorner h \urcorner) := (\omega, 0) \\
    \vartheta(\ulcorner t \urcorner) := (\omega, 0) \\
    \vartheta(\ulcorner R \urcorner) := (\omega, 0) \text{,}
\end{gather*}
where $\mathrm{Pre}$ and $\mathrm{Suc}$ are the predecessor and successor functions on $\omega$, respectively. Set
\begin{equation*}
    \mathfrak{T}(\mathfrak{M}) := (\emptyset, \sigma(\mathfrak{M}), \dot{\mathcal{U}}, \vartheta) \text{.}
\end{equation*}
Intuitively, we want $\ulcorner h \urcorner$ to interpret the current position of the read/write head of $M$, $\ulcorner t \urcorner$ to interpret the current state $M$ is in, and $\ulcorner R \urcorner$ to interpret the current contents of $M$'s tape. Notice that all members of $\omega$ are definable over the membership  relation without parameters. This implies references to them can be made in sentences over $\sigma(\mathfrak{M})$. We will use $\langle n \rangle$ to refer to the $n < \omega$ in such sentences. 

Define 
\begin{align*}
    I(\mathfrak{M}) := \{s : \ & s \models^* \mathfrak{T}(\mathfrak{M}) \text{ and } \\
    & s \models \ulcorner h = \langle 0 \rangle \wedge t = \langle 0 \rangle \urcorner \text{ and } \\
    & s \models \ulcorner \exists x \ (R(x)) \wedge \forall x \ \forall y \ (R(x) \wedge R(y) \implies x = y) \urcorner\} \text{.}
\end{align*}
Each member of $I(\mathfrak{M})$ represents an initial configuration of $M$ with \begin{enumerate}[label=(\alph*)]
    \item the read/write head at the $0$-th cell of the tape, 
    \item the current state being the initial state, 
    \item `$1$' written on exactly one cell of the tape, and
    \item `$0$' written on all other cells of the tape.
\end{enumerate} 
More specifically, `$1$' appears on the $k$-th cell of $M$'s tape in an initial configuration iff $M$ takes $k$ as input.

We are almost done with specifying our SeqA $\mathfrak{M}$; the only thing left to do is to give a transition function $\tau_{\mathfrak{M}}$ accurately representing transitions among the Turing configurations of $M$. It suffices to concoct a suitable sentence $\phi^{\tau}_{\mathfrak{M}}$ based on the transition function $$\delta : \{c_i : i < n-1\} \times 2 \longrightarrow \{c_i : i < n\} \times 2 \times \{\mathrm{left}, \mathrm{right}\}$$ of $M$. For each pair $(j, k) \in n \times 2$, let
\begin{align*}
    \phi_{j,k} := \ & \ulcorner (q(R_0(h_0), k) \wedge t_0 \! = \! \langle j \rangle) \!\! \implies \!\! (q(R_1(h_0), k') \wedge t_1 \! = \! \langle j' \rangle \wedge h_1 \! = \mathrm{Pre}(h_0) \wedge \varphi) \urcorner \\ 
    & \text{ if } j < n-1 \text{ and } \delta((c_j, k)) = (c_{j'}, k', \mathrm{left}) \text{,} \\
    \phi_{j,k} := \ & \ulcorner (q(R_0(h_0), k) \wedge t_0 \! = \! \langle j \rangle) \!\! \implies \!\! (q(R_1(h_0), k') \wedge t_1 \! = \! \langle j' \rangle \wedge h_1 \! = \mathrm{Suc}(h_0) \wedge \varphi) \urcorner \\ 
    & \text{ if } j < n-1 \text{ and } \delta((c_j, k)) = (c_{j'}, k', \mathrm{right}) \text{, and} \\
    \phi_{j,k} := \ & \ulcorner h_0 = h_1 \wedge t_0 = t_1 \wedge \forall x \ (R_0(x) \iff R_1(x)) \urcorner \\ 
    & \text{ if } j = n-1 \text{,}
\end{align*}
where $q(\psi, k)$ is a shorthand for 
\begin{align*}
    \neg \psi & \text{ if } k = 0 \text{;} \\
    \psi & \text{ if } k = 1
\end{align*}
and $\varphi$ is the sentence
\begin{equation*}
    \ulcorner \forall x \ (\neg x = h_0 \implies (R_0(x) \iff R_1(x))) \urcorner \text{.}
\end{equation*}
Finally, 
\begin{equation*}
    \phi^{\tau}_{\mathfrak{M}} := \bigwedge \{\phi_{j,k} : (j, k) \in n \times 2\} \text{,}
\end{equation*}
is a first-order sentence over $2 \cdot \sigma(\mathfrak{M})$ that gives us the representation we want.
\end{ex}

\subsection{Absoluteness and the Bounded Exploration Postulate Revisited}\label{ss240}

A good notion of algorithm ought to be absolute to a certain degree. In set-theoretic contexts, absoluteness for transitive models of set theory comes up the most often. 

\begin{defi}
A sentence $\phi$ in the language of set theory with set parameters is \emph{absolute for transitive models of} $\mathsf{ZFC}$ iff for all $A$ and $B$ such that
\begin{itemize}
    \item $A \subset B$,
    \item $(A; \in)$ and $(B; \in)$ are transitive models of $\mathsf{ZFC}$, and
    \item all parameters of $\phi$ are members of $A$,
\end{itemize}
it must be the case that
\begin{equation*}
    (A; \in) \models \phi \iff (B; \in) \models \phi \text{.}
\end{equation*}
\end{defi}

\begin{defi}
Let $\phi$ be a formula in the language of set theory, with set parameters and free variables $x_1, \dots, x_n$. Suppose $\phi$ defines an $n$-ary relation $R$ (in $V$). Then $\phi$ is an \emph{absolute definition of} (or, $\phi$ \emph{absolutely defines}) $R$ \emph{for transitive models of} $\mathsf{ZFC}$ iff for all $A$ such that
\begin{itemize}
    \item $(A; \in)$ is a transitive model of $\mathsf{ZFC}$, and
    \item all parameters of $\phi$ are members of $A$,
\end{itemize}
it must be the case that
\begin{equation*}
     \{(a_1, \dots, a_n) \in A^n : (A; \in) \models \phi[x_1 \mapsto a_1, \dots, x_n \mapsto a_n]\} = R \cap A \text{.}
\end{equation*}
\end{defi}

Recall that it is implicit in our modified definition of a SeqA $\mathfrak{M}$, that $\phi^{\tau}_{\mathfrak{M}}$ needs to have the following property:
\begin{equation}\label{eq1}
    \forall s_0 \in S(\mathfrak{M}) \ \exists ! s_1 \in S(\mathfrak{M}) \ ((s_0, s_1) \models_2 \phi^{\tau}_{\mathfrak{M}}) \text{,}
\end{equation}
where
\begin{equation*}
    S(\mathfrak{M}) = \{s : s \models^* \mathfrak{T}(\mathfrak{M})\} \text{.}
\end{equation*}
However, given any $\mathfrak{T}(\mathfrak{M})$ and $\phi^{\tau}_{\mathfrak{M}}$, it is not immediate that (\ref{eq1}) is absolute for transitive models of $\mathsf{ZFC}$. As a result, it is difficult to ascertain that the definition of $\tau_{\mathfrak{M}}$ from $\mathfrak{T}(\mathfrak{M})$ and $\phi^{\tau}_{\mathfrak{M}}$ is absolute for transitive models of $\mathsf{ZFC}$. This implores us to make the dependency between the current and the next state of a SeqA even more local.

At this point, it is useful to revisit, in a more precise fashion, Gurevich's bounded exploration postulate, for it is an example of highly local dependency even without the closure of the state space under isomorphisms. Combined with our stipulation that a transition function must be definable, we expect to obtain a level of local definability from which the absoluteness we want can be easily derived. Consequently, abiding by the bounded exploration postulate should net us absoluteness at no additional cost.

\begin{defi}
Let $\mathfrak{M}$ be a SeqA and $s \in S(\mathfrak{M})$. Define $\Delta_{\mathfrak{M}}(s)$ to be the unique function from $\sigma(\mathfrak{M})$ into $\mathcal{P}(b(\mathfrak{M}))$ such that
\begin{enumerate}[label=(\alph*)]
    \item if $\ulcorner X \urcorner \in \sigma(\mathfrak{M})$ is an $n$-ary function symbol, then 
    \begin{equation*}
        \Delta_{\mathfrak{M}}(s)(\ulcorner X \urcorner) = \{\Vec{z}^{\frown}(c) \in b(\mathfrak{M})^{n + 1} : \ulcorner X \urcorner^{s}(\Vec{z}) \neq \ulcorner X \urcorner^{\tau_{\mathfrak{M}}(s)}(\Vec{z}) = c \} \text{,}
    \end{equation*}
    \item if $\ulcorner X \urcorner \in \sigma(\mathfrak{M})$ is an $n$-ary relation symbol, then 
    \begin{equation*}
        \Delta_{\mathfrak{M}}(s)(\ulcorner X \urcorner) = \{\Vec{z} \in b(\mathfrak{M})^n : \neg (\ulcorner X \urcorner^{s}(\Vec{z})) \iff \ulcorner X \urcorner^{\tau_{\mathfrak{M}}(s)}(\Vec{z})\} \text{,}
    \end{equation*}
    \item if $\ulcorner X \urcorner \in \sigma(\mathfrak{M})$ is a constant symbol, then 
    \begin{equation*}
        \Delta_{\mathfrak{M}}(s)(\ulcorner X \urcorner) =
        \begin{cases}
            \emptyset & \text{if } \ulcorner X \urcorner^{s} = \ulcorner X \urcorner^{\tau_{\mathfrak{M}}(s)} \\
            b(\mathfrak{M}) & \text{otherwise.}
        \end{cases}
    \end{equation*}
\end{enumerate}
Note that if $b(\mathfrak{M}) = \emptyset$, then $S(\mathfrak{M})$ is a singleton and for every $s \in S(\mathfrak{M})$, $\tau_{\mathfrak{M}}(s) = s$. Necessarily, any such $\mathfrak{M}$ must have $\Delta_{\mathfrak{M}}(s_0) = \Delta_{\mathfrak{M}}(s_1)$ for all $s_0, s_1 \in S(\mathfrak{M})$.
\end{defi}

\begin{defi}[Gurevich, \cite{gurevich}]\label{def217}
Let $\mathfrak{M}$ be a SeqA. Its transition function $\tau_{\mathfrak{M}}$ satisfies the \emph{bounded exploration postulate} iff there is a finite set of ground terms $\mathcal{D}$ over $\sigma(\mathfrak{M})$ such that for all $\ulcorner X \urcorner \in \sigma(\mathfrak{M})$ and $s_0, s_1 \in S(\mathfrak{M})$,
\begin{equation*}
    \forall t \in \mathcal{D} \ (t^{s_0} = t^{s_1}) \implies \Delta_{\mathfrak{M}}(s_0) = \Delta_{\mathfrak{M}}(s_1) \text{.}
\end{equation*}
\end{defi}

As every relation is a boolean function, there is a canonical way of simulating relation symbols in a ground term. Thus from a computability point of view, the bounded exploration postulate seems intuitive enough at first glance. Upon closer inspection, however, the emphasis on changes in valuations gives rise to artefacts such as the one below.

\begin{ex}\label{ex218}
Consider SeqAs $\mathfrak{M}_1$ and $\mathfrak{M}_2$ such that 
\begin{gather*}
    \sigma(\mathfrak{M}_1) = \sigma(\mathfrak{M}_2) = \{\ulcorner \mathrm{In} \urcorner, \ulcorner \mathrm{Out} \urcorner\} \\
    \mathfrak{T}(\mathfrak{M}_1) = \mathfrak{T}(\mathfrak{M}_2) \text{,}
\end{gather*}
$b(\mathfrak{M}_1)$ is infinite, and
\begin{gather*}
    \phi^{\tau}_{\mathfrak{M}_1} = \ulcorner \forall x \ ((\mathrm{In}_1(x) \iff \neg \mathrm{In}_0(x)) \wedge (\mathrm{Out}_1(x) \iff \neg \mathrm{Out}_0(x))) \urcorner \\
    \phi^{\tau}_{\mathfrak{M}_2} = \ulcorner \forall x \ ((\mathrm{In}_1(x) \iff x \neq x) \wedge (\mathrm{Out}_1(x) \iff \mathrm{Out}_0(x))) \urcorner \text{.}
\end{gather*}
Here, $\ulcorner \mathrm{In} \urcorner$ and $\ulcorner \mathrm{Out} \urcorner$ are supposed to represent input and output tapes respectively, with cells indexed by $b(\mathfrak{M}_1) = b(\mathfrak{M}_2)$ and entries from a binary alphabet. $\tau_{\mathfrak{M}_1}$ flips each bit of the input and output tapes, whereas $\tau_{\mathfrak{M}_2}$ resets the input tape to its all-zeroes state, an operation one can think of as a complete erasure of input tape data.

Observe that $\tau_{\mathfrak{M}_1}$ satisfies the bounded exploration postulate because for all $s_0, s_1 \in S(\mathfrak{M}_1)$, 
\begin{equation*}
    \Delta_{\mathfrak{M}}(s_0) = \Delta_{\mathfrak{M}}(s_1) = \{(\ulcorner \mathrm{In} \urcorner, b(\mathfrak{M}_1)), (\ulcorner \mathrm{Out} \urcorner, b(\mathfrak{M}_1))\} \text{.}
\end{equation*}
On the other hand, $\tau_{\mathfrak{M}_2}$ does not satisfy the bounded exploration postulate because for $s_0, s_1 \in S(\mathfrak{M}_1)$, $\Delta_{\mathfrak{M}}(s_0) = \Delta_{\mathfrak{M}}(s_1)$ only if $\ulcorner \mathrm{In} \urcorner^{s_0} = \ulcorner \mathrm{In} \urcorner^{s_1}$, and whether $\ulcorner \mathrm{In} \urcorner^{s_0} = \ulcorner \mathrm{In} \urcorner^{s_1}$ cannot be determined by just examining a finite fragment of $s_0$ and a finite fragment of $s_1$. This means the transition function of any SeqA built on $\tau_{\mathfrak{M}_2}$ by enlarging its signature while maintaining its function (``erase all data from the input tape and do nothing else''), cannot possibly satisfy the postulate.

We hence find ourselves in a weird position whereby universal bit-flipping --- which requires some information of prior states --- is an acceptably bounded, but universal erasure --- which requires no prior knowledge --- is not.
\end{ex}

Besides the erasure of a tape's contents, one would expect a SeqA $\mathfrak{M}_{\text{copy}}$ which copies the input tape's contents to the output tape, and does nothing else, to only carry out steps as valid as universal bit-flipping. Yet, by an argument similar to that for $\tau_{\mathfrak{M}_2}$ not satisfying the bounded exploration postulate in Example \ref{ex218}, $\mathfrak{M}_{\text{copy}}$ fails to satisfy the same postulate. 

With the aforementioned flaws to \ref{def217} in mind, we seek to reformulate the bounded exploration postulate.

\begin{defi}\label{def220}
Let $\sigma$ be a signature. We say a sentence $\phi$ over $2 \cdot \sigma$ is \emph{bounded for} $\sigma$ iff for each $\ulcorner X \urcorner \in \sigma$ there are $\phi^{\ulcorner X \urcorner}$ and $\psi^{\ulcorner X \urcorner}$ such that
\begin{enumerate}[label=(\alph*)]
    \item $\phi^{\ulcorner X \urcorner}$ is a formula over $\{\ulcorner Y_0 \urcorner : \ulcorner Y \urcorner \in \sigma\}$ ($\subset 2 \cdot \sigma$),
    \item $\psi^{\ulcorner X \urcorner}$ is a sentence over $2 \cdot \sigma$,
    \item if $\ulcorner X \urcorner$ is an $n$-ary function symbol, then 
    \begin{itemize}[label=$\circ$, leftmargin=20pt]
        \item $\phi^{\ulcorner X \urcorner}$ contains $n + 1$ free variables, and
        \item for some variable symbols $y_1, \dots, y_{n+1}$,
        \begin{align*}
            \psi^{\ulcorner X \urcorner} = \ulcorner \forall y_1 \dots \forall y_{n+1} \ (X_1(y_1, \dots, y_n) = y_{n+1} \iff \phi^{\ulcorner X \urcorner}(y_1, \dots, y_{n+1})) \urcorner \text{,}
        \end{align*}
    \end{itemize}
    \item if $\ulcorner X \urcorner$ is an $n$-ary relation symbol, then 
    \begin{itemize}[label=$\circ$, leftmargin=20pt]
        \item $\phi^{\ulcorner X \urcorner}$ contains $n$ free variables, and
        \item for some variable symbols $y_1, \dots, y_n$,
        \begin{align*}
            \psi^{\ulcorner X \urcorner} = \ulcorner \forall y_1 \dots \forall y_n \ (X_1(y_1, \dots, y_n) \iff \phi^{\ulcorner X \urcorner}(y_1, \dots, y_n)) \urcorner \text{,}
        \end{align*}
    \end{itemize}
    \item if $\ulcorner X \urcorner$ is a constant symbol, then 
    \begin{itemize}[label=$\circ$, leftmargin=20pt]
        \item $\phi^{\ulcorner X \urcorner}$ contains $1$ free variable, and
        \item for some variable symbol $y$,
        \begin{align*}
            \psi^{\ulcorner X \urcorner} = \ulcorner \forall y \ (X_1 = y \iff \phi^{\ulcorner X \urcorner}(y)) \urcorner \text{, and}
        \end{align*}
    \end{itemize}
    \item $\phi = \bigwedge \{\psi^{\ulcorner X \urcorner} : \ulcorner X \urcorner \in \sigma\}$.
\end{enumerate}
We call $\phi^{\ulcorner X \urcorner}$ the \emph{witness for} $\ulcorner X \urcorner$ \emph{to} $\phi$ \emph{being bounded for} $\sigma$.
\end{defi}

A few things are easily verifiable from Definition \ref{def220}. Let $\sigma$, $\phi$ and the $\phi^{\ulcorner X \urcorner}$s be as in the definition. Suppose $A$ is a set and $\mathfrak{A}$ is a structure with signature $\sigma$ and base set $A$.
\begin{enumerate}[label=(C\arabic*)]
    \item We can efficiently recover $\phi$ from the $\phi^{\ulcorner X \urcorner}$s, up to substitution of variable symbols.
    \item\label{c2pr} If there exists a structure $\mathfrak{B}$ with signature $\sigma$ and base set $A$, for which $(\mathfrak{A}, \mathfrak{B}) \models_2 \phi$, then such $\mathfrak{B}$ must be unique. Furthermore, $(\mathfrak{A}, \mathfrak{B}) \models_2 \phi$ is absolute for transitive models of $\mathsf{ZFC}$ containing $\mathfrak{A}$ and $\mathfrak{B}$.
    \item\label{c3pr} If there is no structure $\mathfrak{B}$ with signature $\sigma$ and base set $A$, for which $(\mathfrak{A}, \mathfrak{B}) \models_2 \phi$, then the cause of said non-existence is necessarily witnessed by a $\Delta_0$-definable (i.e. definable by a first-order formula over $\{\in\}$ with only bounded quantifiers) fragment of $\mathfrak{A}$'s partial valuation map of one of the $\phi^{\ulcorner X \urcorner}$s, with parameters in $A^{< \omega}$. As $\Delta_0$ definitions are absolute for transitive models of $\mathsf{ZFC}$, the existence of a witness to the non-existence such $\mathfrak{B}$ must be too.
    \item\label{c4pr} Points \ref{c2pr} and \ref{c3pr} guarantee that whenever $\mathfrak{M}$ is a SeqA with $\sigma(\mathfrak{M}) = \sigma$, the statement
    \begin{equation*}
        \forall s_0 \in S(\mathfrak{M}) \ \exists ! s_1 \in S(\mathfrak{M}) \ ((s_0, s_1) \models_2 \phi)
    \end{equation*}
    is absolute for transitive models of $\mathsf{ZFC}$ containing $\mathfrak{T}(\mathfrak{M})$.
    \item Suppose $\phi = \phi^{\tau}_{\mathfrak{M}}$ for some SeqA $\mathfrak{M}$. Then in the context of $\mathfrak{M}$, for each $\ulcorner X \urcorner \in \sigma$, the value of the next state's interpretation of $\ulcorner X \urcorner$ at a point depends only on finitely many values of the current state's interpretations of symbols in $\sigma$. This agrees in spirit with the requirement of bounded exploration, albeit only locally. The caveat of bounded exploration being local is a reasonable --- even expected --- compromise for generalised computation.
\end{enumerate}

\subsection{Input and Output}\label{ss250}

In Example \ref{ex212}, we are able to easily read the input of a Turing machine off its corresponding initial state in the SeqA representation. Similarly, we expect to easily read the output off any of its final states. These are self-evident desiderata if we were to port the input-output paradigm of computability to SeqAs. 

To simplify things, let us require any input or output of a SeqA $\mathfrak{M}$ to be a subset of $b(\mathfrak{M})$.

\begin{defi}
A \emph{loading function} of a SeqA $\mathfrak{M}$ is a function $l$ from $\mathcal{P}(b(\mathfrak{M}))$ onto $I(\mathfrak{M})$, such that for some first-order formula $\phi$ with one free variable $x$ over $\sigma(\mathfrak{M})$,
\begin{align*}
    l(A) = s \text{ iff } \forall a \in b(\mathfrak{M}) \ (a \in A \iff s \models \phi[x \mapsto a])
\end{align*}
whenever $A \in \mathcal{P}(b(\mathfrak{M}))$ and $s \in S(\mathfrak{M})$. In this case, we say $\phi$ is a witness for $l$ (as a loading function of $\mathfrak{M}$).
\end{defi}

\begin{defi}
An \emph{unloading function} of a SeqA $\mathfrak{M}$ is a function $u$ from $S(\mathfrak{M})$ into $\mathcal{P}(b(\mathfrak{M}))$, such that for some first-order formula $\phi$ with one free variable $x$ over $\sigma(\mathfrak{M})$,
\begin{align*}
    u(s) = A \text{ iff } \forall a \in b(\mathfrak{M}) \ (a \in A \iff s \models \phi[x \mapsto a])
\end{align*}
whenever $A \in \mathcal{P}(b(\mathfrak{M}))$ and $s \in S(\mathfrak{M})$. In this case, we say $\phi$ is a witness for $u$ (as an unloading function of $\mathfrak{M}$).
\end{defi}

The next two facts are easy to see.

\begin{fact}\label{fact215}
Let $\mathfrak{M}$ be a SeqA. If $\phi$ is a witness for $l_1$ and $l_2$ as loading functions of $\mathfrak{M}$, then $l_1 = l_2$.
\end{fact}

\begin{fact}\label{fact216}
Let $\mathfrak{M}$ be a SeqA. If $\phi$ is a witness for $u_1$ and $u_2$ as unloading functions of $\mathfrak{M}$, then $u_1 = u_2$.
\end{fact}

Let $\mathfrak{M}$ be a SeqA, $\phi$ be a witness for a loading function of $\mathfrak{M}$, and $\mathfrak{T}(\mathfrak{M}) = (\emptyset, \sigma(\mathfrak{M}), \dot{\mathcal{U}}, \vartheta)$. Assume $|b(\mathfrak{M})| > 1$. 

Replace
\begin{itemize}
    \item each $n$-ary function symbol $\ulcorner f \urcorner \in \sigma(\mathfrak{M})$ with an $(n+1)$-ary relation symbol $\ulcorner R^f \urcorner$, and
    \item each constant symbol $\ulcorner c \urcorner \in \sigma(\mathfrak{M})$ with a unary relation symbol $\ulcorner R^c \urcorner$.
\end{itemize}
Make the same replacements in the definition of $\vartheta$. Fix distinct members of $b(\mathfrak{M})$, $x$ and $y$. Add fresh constant symbols $\ulcorner c \urcorner$, $\ulcorner 0 \urcorner$ and $\ulcorner 1 \urcorner$ and a fresh unary relation symbol $\ulcorner \mathrm{In} \urcorner$ to $\sigma(\mathfrak{M})$, and extend $\vartheta$ so that
\begin{gather*}
    \vartheta(\ulcorner c \urcorner) = (b(\mathfrak{M}), 0) \\
    \vartheta(\ulcorner 0 \urcorner) = (\{x\}, 1) \\
    \vartheta(\ulcorner 1 \urcorner) = (\{y\}, 1) \\
    \vartheta(\ulcorner \mathrm{In} \urcorner) = (b(\mathfrak{M}), 0) \text{.}
\end{gather*}

Now, there is a standard way to speak of functions and constants as if they are relations satisfying either functionality or constancy, wherever relevant. In this way, every first-order sentence over the original $\sigma(\mathfrak{M})$ can be translated into a semantically equivalent sentence over the new $\sigma(\mathfrak{M})$, if we identify each non-relational $\ulcorner X \urcorner$ in the original $\sigma(\mathfrak{M})$ with $\ulcorner R^X \urcorner$. Let $\theta$ be such a translation of $\phi^{\tau}_{\mathfrak{M}}$. Let $\psi$ denote a first-order sentence stating that 
\begin{itemize}
    \item ``$\ulcorner R^f \urcorner$ is a function'' for each function symbol $\ulcorner f \urcorner$ in the original $\sigma(\mathfrak{M})$, and
    \item ``$\ulcorner R^c \urcorner$ is a constant'' for each constant symbol $\ulcorner c \urcorner$ in the original $\sigma(\mathfrak{M})$.
\end{itemize}
Update $\phi^{\tau}_{\mathfrak{M}}$ as follows: 
\begin{align*}
    \phi^{\tau}_{\mathfrak{M}} := \ulcorner & (c_0 = 0_0 \implies (\psi_1 \wedge c_1 = 1_0 \wedge \forall x \ (\mathrm{In}_0 (x) \iff \phi_1 (x) \iff \mathrm{In}_1 (x)))) \\
    & \wedge (\neg c_0 = 0_0 \implies (c_1 = c_0 \wedge \forall x \ (\mathrm{In}_0 (x) \iff \mathrm{In}_1 (x)) \wedge \theta)) \urcorner \text{,}
\end{align*}
where $\psi_1$ and $\phi_1$ are the results of replacing every symbol $\ulcorner z \urcorner \in \sigma(\mathfrak{M})$ occurring in $\psi$ and $\phi$ respectively, with $\ulcorner z_1 \urcorner \in 2 \cdot \sigma(\mathfrak{M})$. 

Next, define 
\begin{equation*}
    \varphi := \bigwedge \{\ulcorner \forall x \ (\neg R(x)) \urcorner : \ulcorner R \urcorner \neq \ulcorner \mathrm{In} \urcorner \text{ and } \ulcorner R \urcorner \text{ is a relation symbol in } \sigma(\mathfrak{M})\} \text{.}
\end{equation*}
Then for every $A \in \mathcal{P}(b(\mathfrak{M}))$ there is a unique $s \in S(\mathfrak{M})$ such that 
\begin{equation*}
    \forall a \ (a \in b(\mathfrak{M}) \implies (a \in A \iff s \models \ulcorner \varphi \wedge c = 0 \wedge \mathrm{In}(x) \urcorner [x \mapsto a])) \text{.}
\end{equation*}
In other words, 
\begin{equation*}
    \phi_l := \ulcorner \varphi \wedge c = 0 \wedge \mathrm{In}(x) \urcorner
\end{equation*}
defines a function $l$ from $\mathcal{P}(b(\mathfrak{M}))$ into $S(\mathfrak{M})$. Resetting $I(\mathfrak{M}) := ran(l)$, we have that $\phi_l$ is a witness for a loading function of the new $\mathfrak{M}$. 

For the new $\mathfrak{M}$, the relation symbol $\ulcorner \mathrm{In} \urcorner$ morally interprets a generalised Turing tape of shape $(b(\mathfrak{M}); \sigma(\mathfrak{M}) \setminus \{\ulcorner \mathrm{In} \urcorner\}$. Each initial state has the input written on the tape, while the other components of the machine are set to default values. As such, variability of initial states occurs only on the tape, as is the case for any Turing machine. 

Identifying each non-relational $\ulcorner X \urcorner$ in the original $\sigma(\mathfrak{M})$ with $\ulcorner R^X \urcorner$, we can conclude that for every $A \subset b(\mathfrak{M})$, 
\begin{enumerate}[label=(\arabic*)]
    \item the original $\mathfrak{M}$ terminates on input $A$ iff the new $\mathfrak{M}$ terminates on input $A$, and
    \item if the original and the new $\mathfrak{M}$ terminate on input $A$ with runs $(s_i : i < m)$ and $(s'_i : i < n)$ respectively, then $n = m + 1$ and for each $i < m$, $s_i$ is a reduct of $s'_{i+1}$.
\end{enumerate} 
This means that our modifications to $\mathfrak{M}$ do not change its essence as an algorithm. There is thus no loss in computability resulting from a Turing-machine-like standardisation of loading functions. A similar (in fact, simpler) argument can be used to justify a Turing-machine-like standardisation of unloading functions.

In accordance to the discussion above, we want a SeqA $\mathfrak{M}$ to have a signature $\sigma(\mathfrak{M})$ containing two distinguished unary relation symbols $\ulcorner \mathrm{In} \urcorner$ and $\ulcorner \mathrm{Out} \urcorner$, representing the input and output tapes respectively. Moreover, we want the specification of $\mathfrak{M}$ to include a first-order sentence $\phi^d_{\mathfrak{M}}$ over $\sigma(\mathfrak{M}) \setminus \{\ulcorner \mathrm{In} \urcorner, \ulcorner \mathrm{Out} \urcorner\}$ that describes the default ``non-tape'' configuration of $\mathfrak{M}$. As in the case of transition functions, we want the loading function described by $\phi^d_{\mathfrak{M}}$ to be absolute for transitive models of $\mathsf{ZFC}$. To do so, we want to first impose a greater amount of structure on the states of a SeqA.

Within the context of Turing computability, a countable set $A$ is computable iff there is a computable subset of $\omega$ coding an isomorphic copy of $A$. By Proposition \ref{prop22}, every set can be similarly coded as a subset of a large enough limit ordinal. It makes sense then to simplify things by requiring the uniform base set of a SeqA $\mathfrak{M}$ to be a limit ordinal $\kappa(\mathfrak{M})$. To successfully decode a set of ordinals using the G\"{o}del pairing function, we also need the true membership relation to be a part of each state of $S(\mathfrak{M})$. In other words, $\ulcorner \in \urcorner \in \sigma(\mathfrak{M})$, and whenever $s \in S(\mathfrak{M})$, $s$ interprets $\ulcorner \in \urcorner$ as the true membership relation so that $s$ an expansion of the structure $(\kappa(\mathfrak{M}); \in)$.

Now we can formally state our requirements for $\phi^d_{\mathfrak{M}}$. Taking a leaf from Definition \ref{def220}, for each 
\begin{equation*}
    \ulcorner X \urcorner \in \sigma(\mathfrak{M}) \setminus \{\ulcorner \in \urcorner, \ulcorner \mathrm{In} \urcorner, \ulcorner \mathrm{Out} \urcorner\}
\end{equation*} there must be $\phi^{\ulcorner X \urcorner}$ and $\psi^{\ulcorner X \urcorner}$ such that
\begin{enumerate}[label=(\alph*)]
    \item $\phi^{\ulcorner X \urcorner}$ is a formula over $\{\ulcorner \in \urcorner\}$,
    \item $\psi^{\ulcorner X \urcorner}$ is a sentence over $\sigma(\mathfrak{M}) \setminus \{\ulcorner \mathrm{In} \urcorner, \ulcorner \mathrm{Out} \urcorner\}$,
    \item if $\ulcorner X \urcorner$ is an $n$-ary function symbol, then 
    \begin{itemize}[label=$\circ$, leftmargin=20pt]
        \item $\phi^{\ulcorner X \urcorner}$ contains $n + 1$ free variables, and
        \item for some variable symbols $y_1, \dots, y_{n+1}$,
        \begin{align*}
            \psi^{\ulcorner X \urcorner} = \ulcorner \forall y_1 \dots \forall y_{n+1} \ (X(y_1, \dots, y_n) = y_{n+1} \iff \phi^{\ulcorner X \urcorner}(y_1, \dots, y_{n+1})) \urcorner \text{,}
        \end{align*}
    \end{itemize}
    \item if $\ulcorner X \urcorner$ is an $n$-ary relation symbol, then 
    \begin{itemize}[label=$\circ$, leftmargin=20pt]
        \item $\phi^{\ulcorner X \urcorner}$ contains $n$ free variables, and
        \item for some variable symbols $y_1, \dots, y_n$,
        \begin{align*}
            \psi^{\ulcorner X \urcorner} = \ulcorner \forall y_1 \dots \forall y_n \ (X(y_1, \dots, y_n) \iff \phi^{\ulcorner X \urcorner}(y_1, \dots, y_n)) \urcorner \text{,}
        \end{align*}
    \end{itemize}
    \item if $\ulcorner X \urcorner$ is a constant symbol, then 
    \begin{itemize}[label=$\circ$, leftmargin=20pt]
        \item $\phi^{\ulcorner X \urcorner}$ contains $1$ free variable, and
        \item for some variable symbol $y$,
        \begin{align*}
            \psi^{\ulcorner X \urcorner} = \ulcorner \forall y \ (X = y \iff \phi^{\ulcorner X \urcorner}(y)) \urcorner \text{, and}
        \end{align*}
    \end{itemize}
    \item $\phi^d_{\mathfrak{M}} = \bigwedge \{\psi^{\ulcorner X \urcorner} : \ulcorner X \urcorner \in \sigma\}$.
\end{enumerate}
We call such $\phi^d_{\mathfrak{M}}$ \emph{simple for} $\sigma(\mathfrak{M})$, and $\phi^{\ulcorner X \urcorner}$ the \emph{witness for} $\ulcorner X \urcorner$ \emph{to} $\phi^d_{\mathfrak{M}}$ \emph{being simple for} $\sigma(\mathfrak{M})$ Additionally,
\begin{equation*}
    \ulcorner \phi^d_{\mathfrak{M}} \wedge \forall y \ (\neg \mathrm{Out}(y)) \wedge \mathrm{In}(x) \urcorner
\end{equation*}
must be a witness for a (necessarily unique, by Fact \ref{fact215}) loading function $l_{\mathfrak{M}}$ of $\mathfrak{M}$. We call $l_{\mathfrak{M}}$ the \emph{default loading function} of $\mathfrak{M}$. The \emph{default unloading function} $u_{\mathfrak{M}}$ of $\mathfrak{M}$ is the unique (by Fact \ref{fact216}) unloading function of $\mathfrak{M}$ for which $\ulcorner \mathrm{Out}(x) \urcorner$ is a witness.

Analogous to \ref{c2pr}, given $A \subset \kappa(\mathfrak{M})$, if there is $s$ satisfying
\begin{equation*}
    \forall a \in \kappa(\mathfrak{M}) \ (a \in A \iff s \models \ulcorner \phi^d_{\mathfrak{M}} \wedge \forall y \ (\neg \mathrm{Out}(y)) \wedge \mathrm{In}(a) \urcorner) \text{,}
\end{equation*}
then such $s$ is unique and absolute for transitive models of $\mathsf{ZFC}$ containing $A$ and $\kappa(\mathfrak{M})$. By an argument similar to that through which \ref{c4pr} is concluded, we can see that 
\begin{equation*}
    \forall A \subset \kappa(\mathfrak{M}) \ \exists ! s \in S(\mathfrak{M}) \ \forall a \in \kappa(\mathfrak{M}) \ (a \in A \iff s \models \ulcorner \phi^d_{\mathfrak{M}} \wedge \forall y \ (\neg \mathrm{Out}(y)) \wedge \mathrm{In}(a) \urcorner)
\end{equation*}
is absolute for transitive models of $\mathsf{ZFC}$ containing $\kappa(\mathfrak{M})$. It goes without saying that a unique $A$ satisfying
\begin{equation*}
    \forall a \in \kappa(\mathfrak{M}) \ (a \in A \iff s \models \ulcorner \mathrm{Out}(a) \urcorner)
\end{equation*}
exists for any given $s \in S(\mathfrak{M})$, and is absolute for transitive models of $\mathsf{ZFC}$ containing $s$ and $\kappa(\mathfrak{M})$. Likewise obvious is the absoluteness of the sentence
\begin{equation*}
    \forall s \in S(\mathfrak{M}) \ \exists ! A \subset \kappa(\mathfrak{M}) \ \forall a \in \kappa(\mathfrak{M}) \ (a \in A \iff s \models \ulcorner \mathrm{Out}(a) \urcorner)
\end{equation*}
for transitive models of $\mathsf{ZFC}$. As a consequence, our designated witnesses for the default loading and unloading functions must also be absolute for transitive models of $\mathsf{ZFC}$.

\subsection{Constraining the Interpretation Constraints} 

Now that we require $\ulcorner \in \urcorner$ to be a member of the signature of every SeqA, and for it to always be interpreted as the true membership relation, should there be further requirements on the interpretation constraints of $\mathfrak{T}(\mathfrak{M})$?

\begin{ex}\label{ex217}
Choose a limit ordinal $\kappa$ and $A \subset \kappa$. Let $\mathfrak{M}$ be a SeqA with 
\begin{gather*}
    \kappa(\mathfrak{M}) = \kappa \\
    \sigma(\mathfrak{M}) = \{\ulcorner \in \urcorner, \ulcorner \mathrm{In} \urcorner, \ulcorner \mathrm{Out} \urcorner, \ulcorner R \urcorner\}
\end{gather*} 
and associated TCI 
\begin{equation*}
    \mathfrak{T}(\mathfrak{M}) = (\emptyset, \sigma(\mathfrak{M}), \dot{\mathcal{U}}, \vartheta) \text{,}
\end{equation*} 
such that $\ulcorner \in \urcorner$, $\ulcorner \mathrm{In} \urcorner$ and $\ulcorner \mathrm{Out} \urcorner$ are as stipulated and $\ulcorner R \urcorner$ is a unary relation symbol. To wit, it must be the case that
\begin{gather*}
    \vartheta(\dot{\mathcal{U}}) = (\kappa, 1) \\
    \vartheta(\ulcorner \in \urcorner) = (\in \cap \ \kappa, 1) \\
    \vartheta(\ulcorner \mathrm{In} \urcorner) = (\kappa, 0) \\
    \vartheta(\ulcorner \mathrm{Out} \urcorner) = (\kappa, 0) \text{.}
\end{gather*}
Now suppose
\begin{gather*}
    \vartheta(\ulcorner R \urcorner) = (A, 1) \\
    \phi^d_{\mathfrak{M}} = \ulcorner \urcorner
\end{gather*}
and
\begin{align*}
    \phi^{\tau}_{\mathfrak{M}} = \ulcorner & (\forall x \ (R_0(x) \iff \mathrm{Out}_0(x)) \implies \\
    & \forall x \ ((\mathrm{In}_0(x) \iff \mathrm{In}_1(x)) \wedge (\mathrm{Out}_0(x) \iff \mathrm{Out}_1(x)))) \\
    & \wedge (\neg (\forall x \ (R_0(x) \iff \mathrm{Out}_0(x))) \implies \\
    & \forall x \ ((\mathrm{In}_0(x) \iff \mathrm{In}_1(x)) \wedge (R_0(x) \iff \mathrm{Out}_1(x)))) \urcorner \text{.}
\end{align*}
Then basically, $\mathfrak{M}$ is an algorithm that starts with $\ulcorner R \urcorner$ holding $A$, before copying $A$ onto the output tape. As it does this regardless of input, $\mathfrak{M}$ computes $A$ from any input, even if $A$ is a very complicated set by any other measure.
\end{ex}

In the usual input-output paradigm of computability where the output complexity should depend solely on the input complexity, Example \ref{ex217} cautions against hiding too much information in the states of a SeqA. One way to fix this is to enforce certain definability constraints on the interpretation constraints. 

Fix a SeqA $\mathfrak{M}$ with $\mathfrak{T}(\mathfrak{M}) = (\emptyset, \sigma(\mathfrak{M}), \dot{\mathcal{U}}, \vartheta)$. We want $A$ to be definable over the structure $(\kappa(\mathfrak{M}); \in)$ for all $(A, z) \in ran(\vartheta)$. In other words, if $(A, z) \in ran(\vartheta)$, then for some $n < \omega$ there is a formula $\varphi$ over $\{\ulcorner \in \urcorner\}$ with free variables $x_1, \dots, x_n$ such that 
\begin{equation*}
    A = \{(\alpha_1, \dots, \alpha_n) \in \kappa^n : (\kappa; \in) \models \varphi[x_1 \mapsto \alpha_n, \dots, x_n \mapsto \alpha_n]\} \text{.}
\end{equation*}
Due to the fact that every state of $\mathfrak{M}$ is an expansion of $(\kappa(\mathfrak{M}); \in)$, we can actually include the definition of any such $A$ in $\phi^{\tau}_{\mathfrak{M}}$ --- more specifically, in each $\phi^{\ulcorner X \urcorner}$ for which $\vartheta(\ulcorner X \urcorner) = (A, z)$. Doing this would allow all the ``programming'' to occur in the formulation of $\phi^{\tau}_{\mathfrak{M}}$; all symbols in $\sigma(\mathfrak{M})$ besides $\ulcorner \in \urcorner$ can simply be viewed as declarations of the ``programming variables'' we expect to occur in $\phi^{\tau}_{\mathfrak{M}}$. 

Justifiably, we shall stipulate that whenever $(A, z) \in ran(\vartheta)$, $z = 0$ and $A = \kappa^n$ for some $n < \omega$. We are now ready to formally define our modified version of a SeqA.

\begin{defi}\label{def224}
A \emph{generalised sequential algorithm} --- henceforth \textit{GSeqA} --- $\mathfrak{M}$ is fully specified by a limit ordinal $\kappa(\mathfrak{M})$, a finite signature $\sigma(\mathfrak{M})$, and sentences $\phi^{\tau}_{\mathfrak{M}}$ and $\phi^d_{\mathfrak{M}}$ such that
\begin{enumerate}[label=(D\arabic*)]
    \item $\{\ulcorner \in \urcorner, \ulcorner \mathrm{In} \urcorner, \ulcorner \mathrm{Out} \urcorner\} \subset \sigma(\mathfrak{M})$,
    \item there is a unique $\vartheta_{\mathfrak{M}} : (\sigma(\mathfrak{M}) \cup \{\dot{\mathcal{U}}\}) \longrightarrow V$ with the following properties:
    \begin{itemize}
        \item $\vartheta(\dot{\mathcal{U}}) = (\kappa(\mathfrak{M}), 1)$,
        \item $\vartheta(\ulcorner \in \urcorner) = (\in, 1)$,
        \item $\vartheta(\ulcorner \mathrm{In} \urcorner) = (\kappa(\mathfrak{M}), 0)$,
        \item $\vartheta(\ulcorner \mathrm{Out} \urcorner) = (\kappa(\mathfrak{M}), 0)$, and
        \item for all $\ulcorner X \urcorner \in \sigma(\mathfrak{M}) \setminus \{\ulcorner \in \urcorner, \ulcorner \mathrm{In} \urcorner, \ulcorner \mathrm{Out} \urcorner\}$ there is $n < \omega$ satisfying 
        \begin{equation*}
            \vartheta(\ulcorner X \urcorner) = (\kappa(\mathfrak{M})^n, 0) \text{,}
        \end{equation*}
    \end{itemize}
    \item there is a unique $\mathfrak{T}(\mathfrak{M})$ for which $\mathfrak{T}(\mathfrak{M}) = (\emptyset, \sigma(\mathfrak{M}), \dot{\mathcal{U}}, \vartheta_{\mathfrak{M}})$,
    \item there is a unique $S(\mathfrak{M})$ for which $S(\mathfrak{M}) = \{s : s \models^* \mathfrak{T}(\mathfrak{M})\}$,
    \item $\phi^{\tau}_{\mathfrak{M}}$ is bounded for $\sigma(\mathfrak{M})$, 
    \item for all $s_0 \in S(\mathfrak{M})$ there is a unique $s_1 \in S(\mathfrak{M})$ for which $(s_0, s_1) \models_2 \phi^{\tau}_{\mathfrak{M}}$,
    \item $\phi^d_{\mathfrak{M}}$ is simple for $\sigma(\mathfrak{M})$, and
    \item for all $A \subset \kappa(\mathfrak{M})$ there is a unique $s \in S(\mathfrak{M})$ for which
    \begin{equation*}
        \forall a \in \kappa(\mathfrak{M}) \ (a \in A \iff s \models \ulcorner \phi^d_{\mathfrak{M}} \wedge \forall y \ (\neg \mathrm{Out}(y)) \wedge \mathrm{In}(a) \urcorner) \text{.}
    \end{equation*}
\end{enumerate}
In this case, we say $\mathfrak{M} = (\kappa(\mathfrak{M}), \sigma(\mathfrak{M}), \phi^{\tau}_{\mathfrak{M}}, \phi^d_{\mathfrak{M}})$.
\end{defi}

\begin{rem}\label{rem225}
Fix any limit ordinal $\kappa$. Notice that there are only countably many GSeqAs $\mathfrak{M}$ satisfying $\kappa(\mathfrak{M}) = \kappa$. Indeed, this countable set of GSeqAs specifications is absolute for transitive models of $\mathsf{ZFC}$ containing $\kappa$.
\end{rem}

\begin{defi}
If $\mathfrak{M} = (\kappa(\mathfrak{M}), \sigma(\mathfrak{M}), \phi^{\tau}_{\mathfrak{M}}, \phi^d_{\mathfrak{M}})$ is a GSeqA, define 
\begin{enumerate}[label=(\arabic*)]
    \item $S(\mathfrak{M})$ as in Definition \ref{def224},
    \item $\tau_{\mathfrak{M}} := \{(s_0, s_1) \in S(\mathfrak{M})^2 : (s_0, s_1) \models_2 \phi^{\tau}_{\mathfrak{M}}\}$,
    \item 
    \!
    $\begin{aligned}[t]
       l_{\mathfrak{M}} := \{ & (A, s) \in \mathcal{P}(\kappa(\mathfrak{M})) \times S(\mathfrak{M}) : \\
       & \forall a \in \kappa(\mathfrak{M}) \ (a \in A \iff s \models \ulcorner \phi^d_{\mathfrak{M}} \wedge \forall y \ (\neg \mathrm{Out}(y)) \wedge \mathrm{In}(a) \urcorner)\} 
    \end{aligned}$
    \item $I(\mathfrak{M}) := ran(l_{\mathfrak{M}})$
    \item 
    \!
    $\begin{aligned}[t]
        u_{\mathfrak{M}} := \{ & (s, A) \in S(\mathfrak{M}) \times \mathcal{P}(\kappa(\mathfrak{M})) : \\
        & \forall a \in \kappa(\mathfrak{M}) \ (a \in A \iff s \models \ulcorner \mathrm{Out}(a) \urcorner)\} \text{.}
    \end{aligned}$
\end{enumerate}
\end{defi}

The next proposition highlights the important properties of GSeqAs, based on what we have discussed so far.

\begin{prop}\label{prop226}
Let $\mathfrak{M}$ be a GSeqA. Then the following are absolute for transitive models of $\mathsf{ZFC}$ containing $\mathfrak{T}(\mathfrak{M})$:
\begin{enumerate}[label=(\arabic*)]
    \item the definition of $S(\mathfrak{M})$,
    \item the definition of $\tau_{\mathfrak{M}}$ and $\tau_{\mathfrak{M}}$ being a function from $S(\mathfrak{M})$ into $S(\mathfrak{M})$,
    \item the definition of $l_{\mathfrak{M}}$ and $l_{\mathfrak{M}}$ being a function from $\mathcal{P}(\kappa(\mathfrak{M}))$ into $S(\mathfrak{M})$,
    \item the definition of $I(\mathfrak{M})$,
    \item the definition of $u_{\mathfrak{M}}$ and $u_{\mathfrak{M}}$ being a function from $S(\mathfrak{M})$ into $\mathcal{P}(\kappa(\mathfrak{M}))$. 
\end{enumerate}
\end{prop}

Now that GSeqAs have been officially defined, it is imperative to revisit Example \ref{ex212}, for any notion of generalised algorithm that cannot represent an arbitrary Turing cannot be too useful.

\begin{ex}[\ref{ex212} revisited]\label{ex228}
Let us describe changes to the $\mathfrak{M}$ defined in Example \ref{ex212} that would qualify it as a GSeqA. First, an observation: there are formulas $\phi_P$ and $\phi_S$ over $\ulcorner \in \urcorner$, each with two free variables $x$ and $y$, such that for all $s \in S(\mathfrak{M})$,
\begin{equation*}
    s \models \forall x \ \forall y \ ((\mathrm{Pre}(x) = y \iff \phi_P(x, y)) \wedge (\mathrm{Suc}(x) = y \iff \phi_S(x, y))) \text{.}
\end{equation*}
Since the definition of the functions $\mathrm{Pre}$ and $\mathrm{Suc}$ are absolute for transitive models of $\mathsf{ZFC}$, we can omit $\ulcorner \mathrm{Pre} \urcorner$ and $\ulcorner \mathrm{Suc} \urcorner$ from $\sigma(\mathrm{M})$.

Next, we shall replace $\ulcorner R \urcorner$ with $\ulcorner \mathrm{Out} \urcorner$ while adding $\ulcorner \mathrm{In} \urcorner$ to $\sigma(\mathrm{M})$. We also want to have another constant symbol $\ulcorner e \urcorner$ in $\sigma(\mathrm{M})$. Set 
\begin{gather*}
    \vartheta_{\mathfrak{M}}(\ulcorner e \urcorner) := (\omega, 0) \\
    \phi^d_{\mathfrak{M}} := \ulcorner h = \langle 0 \rangle \wedge t = \langle 0 \rangle \wedge e = \langle 0 \rangle \urcorner \text{.}
\end{gather*}

We are left to define the $\phi^{\ulcorner X \urcorner}$s for $\ulcorner X \urcorner \in \sigma(\mathfrak{M})$. The only non-trivial ones are for $\ulcorner X \urcorner \in \{\ulcorner h \urcorner, \ulcorner t \urcorner, \ulcorner e \urcorner, \ulcorner \mathrm{In} \urcorner, \ulcorner \mathrm{Out} \urcorner\}$. So let
\begin{align*}
    \phi^{\ulcorner h \urcorner}(x) := \ulcorner & (e_0 = \langle 0 \rangle \wedge h_0 = x) \ \vee \\
    & \bigvee \{(\neg e_0 = \langle 0 \rangle) \wedge q(\mathrm{Out}_0(h_0), k) \wedge t_0 \! = \! \langle j \rangle \wedge x = \mathrm{Pre}(h_0) : \\
    & \exists x \ \exists y \ (\delta((c_j, k)) = (x, y, \mathrm{left}))\} \ \vee \\ 
    & \bigvee \{(\neg e_0 = \langle 0 \rangle) \wedge q(\mathrm{Out}_0(h_0), k) \wedge t_0 \! = \! \langle j \rangle \wedge x = \mathrm{Suc}(h_0) : \\
    & \exists x \ \exists y \ (\delta((c_j, k)) = (x, y, \mathrm{right}))\} \urcorner
\end{align*}
\begin{align*}
    \phi^{\ulcorner t \urcorner}(x) := \ulcorner & (e_0 = \langle 0 \rangle \wedge t_0 = x) \ \vee \\
    & \bigvee \{(\neg e_0 = \langle 0 \rangle) \wedge q(\mathrm{Out}_0(h_0), k) \wedge t_0 \! = \! \langle j \rangle \wedge x = \langle j' \rangle : \\
    & \exists x \ \exists y \ (\delta((c_j, k)) = (c_{j'}, x, y))\} \urcorner \\
    \phi^{\ulcorner e \urcorner}(x) := \ulcorner & (e_0 = \langle 0 \rangle \wedge x = \langle 1 \rangle) \vee ((\neg e_0 = \langle 0 \rangle) \wedge x = e_0) \urcorner \\
    \phi^{\ulcorner \mathrm{In} \urcorner}(x) := \ulcorner & \mathrm{In}_0(x) \urcorner \\
    \phi^{\ulcorner \mathrm{Out} \urcorner}(x) := \ulcorner & (e_0 = \langle 0 \rangle \wedge \mathrm{In}_0(x)) \vee ((\neg e_0 = \langle 0 \rangle) \wedge (\neg x = h_0) \wedge \mathrm{Out}_0(x)) \ \vee \\
    & \bigvee \{(\neg e_0 = \langle 0 \rangle) \wedge q(\mathrm{Out}_0(h_0), k) \wedge t_0 \! = \! \langle j \rangle \wedge x = h_0 \wedge (\neg x = x) : \\
    & \exists x \ \exists y \ (\delta((c_j, k)) = (x, 0, y))\} \ \vee \\
    & \bigvee \{(\neg e_0 = \langle 0 \rangle) \wedge q(\mathrm{Out}_0(h_0), k) \wedge t_0 \! = \! \langle j \rangle \wedge x = h_0 \wedge x = x : \\
    & \exists x \ \exists y \ (\delta((c_j, k)) = (x, 1, y))\} \urcorner \text{,}
\end{align*}
where $x = \mathrm{Pre}(h_0)$ and $x = \mathrm{Suc}(h_0)$ are shorthand notations for semantically equivalent formulas over $\{\ulcorner \in_0 \urcorner, \ulcorner h_0 \urcorner\}$. 

A routine check would reveal that the aforementioned details are sufficient to specify $\mathfrak{M}$ as a GSeqA, and that they allow $\mathfrak{M}$ to accurately simulate the Turing machine $M$. 
\end{ex}

From the next subsection onwards, the construction of GSeqAs used in proofs and arguments can be too cumbersome to formally spell out. In such cases, only high-level details will be given. The reader should always try to convince themselves that those details can be properly realised by GSeqAs.

\subsection{Transfinite Runs}\label{ss270}

Consider a finite terminating run $\Vec{s} = (s_1, \dots, s_n)$ of a GSeqA $\mathfrak{M}$. Then by combining witnesses to $\phi^{\tau}_{\mathfrak{M}}$ being bounded for $\sigma(\mathfrak{M})$, we can conjure witnesses to some sentence $\phi$ being bounded for $\sigma(\mathfrak{M})$, such that using $\phi$ in place of $\phi^{\tau}_{\mathfrak{M}}$ in the definition of $\mathfrak{M}$ gives us a GSeqA $\mathfrak{M}'$ with $\tau_{\mathfrak{M}'}(s_1) = s_n$. Therefore, as far as computability is concerned, allowing finite steps of arbitrary sizes in computations is no more powerful than allowing only one-step computations. This goes to show that
GSeqAs are in fact built for runs of infinite lengths.

Given an infinite ordinal $\delta$, what does it mean to have a GSeqA computation of length $\delta$? Naturally, every state of the computation must depend on the states prior. For successor ordinals $\alpha < \delta$, the $\alpha$-th state of the computation depends solely on its unique previous state, according to the transition function of the GSeqA. But what about limit ordinals less than $\delta$? How do we describe the states associated with these ordinals in terms of states occurring in the computation before them? This is equivalent to asking how we are going to process the limit of a sequence of states.

Fix a GSeqA $\mathfrak{M}$, a limit ordinal $\delta$, and let $Sq$ be a sequence of states of $\mathfrak{M}$. We want to have a uniform definition of $lim(Sq) \in S(\mathfrak{M})$. One natural way to do so is to take ordinal limits point-wise (according to the order topology on $\kappa(\mathfrak{M})$ induced by $\in$), while treating relations as boolean functions. The obvious hindrance to this method is, the limit of a sequence of ordinals does not always exist. However, for various purposes, there are suitable ersatzes.

\begin{defi}
Let $T$ be a function from some ordinal into $ORD$. Then $lim \ inf (T)$, $lim \ sup (T)$ and $lim(T)$ refer respectively to the limit inferior, the limit superior and the limit of $T$ as a sequence, with respect to the order topology induced on $\kappa(\mathfrak{M})$ by $\in$.
\end{defi}

\begin{fact}
If $T$ is a function from some ordinal into $ORD$, then $lim \ inf (T)$ and $lim \ sup (T)$ always exist.
\end{fact}

\begin{fact}\label{fact230}
If $T$ is a function from $\delta < \kappa(\mathfrak{M})$ into $ORD$ and
\begin{equation*}
    sup \{T(\alpha) : \alpha \in \delta\} < \kappa(\mathfrak{M}) \text{,}
\end{equation*}
then $lim \ inf (T)$ and $lim \ sup (T)$ always exist in $\kappa(\mathfrak{M})$.
\end{fact}

Point-wise $lim \ inf$ and $lim \ sup$ operations are commonly used to define the limit stages of transfinite computations. Hamkins and Lewis, in \cite{ittm}, defined each cell in the tape of an infinite time Turing machine to take the $lim \ sup$ of its previous values. Koepke, in his definition of ordinal Turing machines (\cite{koepke1} and \cite{koepke2}), instead require each cell of the machine tape to take the $lim \ inf$ of its previous values. Point-wise $lim \ inf$ and $lim \ sup$ always exist in these generalised Turing machines, for computations shorter than the lengths of their tapes. This is because single-step movements of these machines are always small, and consequently the hypothesis of Fact \ref{fact230} applies.

Hereon, we shall identify relations with boolean functions and constants with $0$-ary functions wherever convenient, in part so that only function symbols need to be considered in definitions and case analyses. 

\begin{defi}
Let $\mathfrak{M}$ be a GSeqA and $Sq$ be a function from some ordinal into $S(\mathfrak{M})$. Given $\ulcorner X \urcorner \in \sigma(\mathfrak{M})$ an $n$-ary function symbol, denote by $\ulcorner R^{\ulcorner X \urcorner} \urcorner$ a distinguished $(n+2)$-ary relation symbol and define the relation 
\begin{equation*}
    R_{Sq}^{\ulcorner X \urcorner} := \{(\alpha, \beta_1, \dots, \beta_n, \beta) \in dom(Sq) \times \kappa(\mathfrak{M})^{n+1} : X^{Sq(\alpha)} (\beta_1, \dots, \beta_n) = \beta\} \text{.}
\end{equation*}
\end{defi}

\begin{defi}
Let $\mathfrak{M}$ be a GSeqA. Given $\ulcorner X \urcorner \in \sigma(\mathfrak{M})$ an $n$-ary function symbol, we say $\phi$ is a \emph{simple limit formula for} $(\ulcorner X \urcorner, \mathfrak{M})$ iff $\phi$ is a formula over $\{\ulcorner \in \urcorner, \ulcorner R^{\ulcorner X \urcorner} \urcorner\}$ with free variables $x_1, \dots, x_{n+1}$.
\end{defi}

\begin{defi}\label{def234c}
Let $\mathfrak{M}$ be a GSeqA. Given $\ulcorner X \urcorner \in \sigma(\mathfrak{M})$ an $n$-ary function symbol and $\phi$ a \emph{simple limit formula for} $(\ulcorner X \urcorner, \mathfrak{M})$, we say $\phi$ \emph{preserves ordinal limits} iff interpreting $\ulcorner \in \urcorner$ as $\in$ and $\ulcorner R^{\ulcorner X \urcorner} \urcorner$ as $R_{Sq}^{\ulcorner X \urcorner}$, 
\begin{equation*}
    (\kappa(\mathfrak{M}); \in, R_{Sq}^{\ulcorner X \urcorner}) \models \phi[x_1 \mapsto \alpha_1, \dots, x_{n+1} \mapsto \alpha_{n+1}]
\end{equation*}
for all functions $Sq$ from a limit ordinal $< \kappa(\mathfrak{M})$ into $S(\mathfrak{M})$ and all $(\alpha_1, \dots, \alpha_{n+1}) \in \kappa(\mathfrak{M})^{n+1}$ such that 
\begin{equation*}
    \alpha_{n+1} = lim \{(\beta, \ulcorner X \urcorner^{Sq(\beta)}(\alpha_1, \dots, \alpha_n)) : \beta < dom(Sq)\} \text{.}
\end{equation*}
\end{defi}

\begin{defi}\label{def234}
Let $\mathfrak{M}$ be a GSeqA and $Sq$ be a function from a limit ordinal $< \kappa(\mathfrak{M})$ into $S(\mathfrak{M})$. A state $s \in S(\mathfrak{M})$ is a \emph{simple limit of} $Sq$ iff for each $n$-ary function symbol $\ulcorner X \urcorner \in \sigma(\mathfrak{M})$ there is a simple limit formula $\phi_{lim}^{\ulcorner X \urcorner}$ for $(\ulcorner X \urcorner, \mathfrak{M})$ such that, interpreting $\ulcorner \in \urcorner$ as $\in$ and $\ulcorner R^{\ulcorner X \urcorner} \urcorner$ as $R_{Sq}^{\ulcorner X \urcorner}$, 
\begin{equation*}
    \ulcorner X \urcorner^s (\alpha_1, \dots, \alpha_n) = \alpha_{n+1} \iff (\kappa(\mathfrak{M}); \in, R_{Sq}^{\ulcorner X \urcorner}) \models \phi[x_1 \mapsto \alpha_1, \dots, x_{n+1} \mapsto \alpha_{n+1}]
\end{equation*}
whenever $(\alpha_1, \dots, \alpha_{n+1}) \in \kappa(\mathfrak{M})^{n+1}$. In this case, we say $\{\phi_{lim}^{\ulcorner X \urcorner} : \ulcorner X \urcorner \in \sigma(\mathfrak{M})\}$ \emph{witnesses} $s$ \emph{is a simple limit of} $Sq$.
\end{defi}

It is immediate from Definition \ref{def234} and our assumption that $\sigma(\mathfrak{M})$ contains only function symbols, that a simple limit of a $< \kappa(\mathfrak{M})$-length ordinal sequence of states of $\mathfrak{M}$ must be unique if it exists. Furthermore, for any such sequence $Sq$ and any set $Y$, whether $Y$ witnesses the existence of a simple limit of $Sq$, as well as the identity of said limit if it exists, is absolute for transitive models of $\mathsf{ZFC}$ containing $Sq$.

\begin{prop}\label{prop233}
Let $\mathfrak{M}$ be a GSeqA. For each $n$-ary function symbol $\ulcorner X \urcorner \in \sigma(\mathfrak{M})$ there is simple limit formula $\phi_{lim}^{\ulcorner X \urcorner}$ for $(\ulcorner X \urcorner, \mathfrak{M})$ that preserves ordinal limits, such that whenever $Sq$ is a function from $\delta < \kappa(\mathfrak{M})$ into $S(\mathfrak{M})$, $\{\phi_{lim}^{\ulcorner X \urcorner} : \ulcorner X \urcorner \in \sigma(\mathfrak{M})\}$ witnesses the existence of a simple limit $s$ of $Sq$ satisfying 
\begin{gather*}
    \ulcorner X \urcorner^{s} (\alpha_1, ..., \alpha_n) = 
    \begin{cases}
        lim \ inf \{(\beta, \ulcorner X \urcorner^{Sq(\beta)}(\alpha_1, ..., \alpha_n)) : \beta < \delta\} & \!\!\!\! \text{if it exists} < \kappa(\mathfrak{M}) \\
        0 & \!\!\!\! \text{otherwise}
    \end{cases} 
\end{gather*}
for every $n$-ary function symbol $\ulcorner X \urcorner \in \sigma(\mathfrak{M})$.
\end{prop}

\begin{proof}
By the definition of a GSeqA and the observation that 
\begin{equation*}
    lim \ inf (T) = lim (T)
\end{equation*}
whenever $T$ is a function from some ordinal in $ORD$ and $lim (T)$ exists. The reader should check that the $lim \ inf$ values specified by the proposition relative to $\ulcorner X \urcorner$ and $Sq$, in case they exist, are definable over $(\kappa(\mathfrak{M}); \in, R_{Sq}^{\ulcorner X \urcorner})$. In fact, there is one such definition that depends only on $\ulcorner X \urcorner$ and not on $\mathfrak{M}$.
\end{proof}

Colloquially, Proposition \ref{prop233} tells us that every GSeqA is \emph{strongly closed under limit inferiors}. The following definitions generalise this concept.

\begin{defi}\label{def237}
Let $\mathfrak{M}$ be a GSeqA. Then $\mathfrak{M}$ is \emph{strongly closed under} a set $Y$ iff for each $n$-ary function symbol $\ulcorner X \urcorner \in \sigma(\mathfrak{M})$ there is a simple limit formula $\phi_{lim}^{\ulcorner X \urcorner}$ for $(\ulcorner X \urcorner, \mathfrak{M})$, such that 
\begin{equation*}
    Y = \{\phi_{lim}^{\ulcorner X \urcorner} : \ulcorner X \urcorner \in \sigma(\mathfrak{M})\}
\end{equation*}
and whenever $Sq$ is a function from some limit ordinal $< \kappa(\mathfrak{M})$ into $S(\mathfrak{M})$, there is $s \in S(\mathfrak{M})$ witnessed by $Y$ to be a simple limit of $Sq$.
\end{defi}

\begin{defi}\label{def238}
Let $\mathfrak{M}$ be a GSeqA. Then $\mathfrak{M}$ is \emph{practically closed under} a set $Y$ iff for each $n$-ary function symbol $\ulcorner X \urcorner \in \sigma(\mathfrak{M})$ there is a simple limit formula $\phi_{lim}^{\ulcorner X \urcorner}$ for $(\ulcorner X \urcorner, \mathfrak{M})$, such that 
\begin{equation*}
    Y = \{\phi_{lim}^{\ulcorner X \urcorner} : \ulcorner X \urcorner \in \sigma(\mathfrak{M})\}
\end{equation*}
and whenever $Sq$ is a function from some limit ordinal $\delta < \kappa(\mathfrak{M})$ into $S(\mathfrak{M})$ satisfying
\begin{itemize}
    \item $Sq(0) \in I(\mathfrak{M})$,
    \item $Sq(\alpha + 1) = \tau_{\mathfrak{M}}(Sq(\alpha))$ for all $\alpha < \delta$, and
    \item $Y$ witnesses $Sq(\gamma)$ is a simple limit of $Sq \restriction \gamma$ for all limit $\gamma < \delta$,
\end{itemize}
there is $s \in S(\mathfrak{M})$ witnessed by $Y$ to be a simple limit of $Sq$.
\end{defi}

Clearly, for any GSeqA $\mathfrak{M}$ and any set $Y$,
\begin{align*}
    \mathfrak{M} \text{ is strongly closed under } Y \implies \mathfrak{M} \text{ is practically closed under } Y \text{.}
\end{align*}

We are ready to define transfinite terminating runs of GSeqAs.

\begin{defi}
Let $\mathfrak{M}$ be a GSeqA that is practically closed under $Y$. A \emph{short terminating run of} $(\mathfrak{M}, Y)$ \emph{with input} $A$ \emph{and output} $B$ is a function $Sq$ from some successor ordinal $\delta < \kappa(\mathfrak{M})$ into $S(\mathfrak{M})$ such that 
\begin{enumerate}[label=(\alph*)]
    \item $A, B \subset \kappa(\mathfrak{M})$,
    \item $Sq(0) = l_{\mathfrak{M}}(A)$,
    \item $Sq(\alpha + 1) = \tau_{\mathfrak{M}}(Sq(\alpha)) \neq Sq(\alpha)$ for all $\alpha < \delta - 1$,
    \item $Y$ witnesses $Sq(\gamma)$ is a simple limit of $Sq \restriction \gamma$ for all limit $\gamma < \delta$,
    \item $\tau_{\mathfrak{M}}(Sq(\delta - 1)) = Sq(\delta - 1)$, and
    \item $u_{\mathfrak{M}}(Sq(\delta - 1)) = B$.
\end{enumerate}
\end{defi}

\begin{prop}\label{prop239}
Suppose $\mathfrak{M}$ is a GSeqA practically closed under $Y$. Then for some signature $\sigma'$, as long as $\mathfrak{M}'$ is a GSeqA satisfying 
\begin{itemize}
    \item $\sigma(\mathfrak{M}') = \sigma'$ and
    \item $\kappa(\mathfrak{M}') = \kappa(\mathfrak{M})$,
\end{itemize}
and $Z$ is a set of simple limit formulas preserving ordinal limits with $\mathfrak{M}'$ strongly closed under $Z$, there is a GSeqA $\mathfrak{M}^*$ such that
\begin{enumerate}[label=(\alph*)]
    \item\label{2391} $\sigma(\mathfrak{M}^*) = \sigma'$,
    \item\label{2392} $\kappa(\mathfrak{M}^*) = \kappa(\mathfrak{M})$, and
    \item\label{2393} for all $A, B \subset \kappa(\mathfrak{M})$, a short terminating run of $(\mathfrak{M}, Y)$ with input $A$ and output $B$ exists iff a short terminating run of $(\mathfrak{M}^*, Z)$ with input $A$ and output $B$ exists.
\end{enumerate}
\end{prop}

\begin{proof}
Define 
\begin{equation*}
    \sigma' := \sigma(\mathfrak{M}) \cup \{\ulcorner R^{\ulcorner X \urcorner} \urcorner : \ulcorner X \urcorner \in \sigma(\mathfrak{M})\} \cup \{\ulcorner c_0 \urcorner, \ulcorner c_1 \urcorner\} \text{,}
\end{equation*}
where $\ulcorner c_0 \urcorner, \ulcorner c_1 \urcorner$ are constant symbols distinct from all members of
\begin{equation*}
    \sigma(\mathfrak{M}) \cup \{\ulcorner R^{\ulcorner X \urcorner} \urcorner : \ulcorner X \urcorner \in \sigma(\mathfrak{M})\} \text{.}
\end{equation*}

Let $\mathfrak{M}'$ and $Z$ be as given by the proposition. We need to define $\mathfrak{M}^*$ so that \ref{2391} to \ref{2393} are satisfied. We are thus forced to set 
\begin{gather*}
    \sigma(\mathfrak{M}^*) := \sigma' \\
    \kappa(\mathfrak{M}^*) := \kappa \text{.}
\end{gather*}
But this means $S(\mathfrak{M}^*) = S(\mathfrak{M}')$, so $\mathfrak{M}^*$ is strongly closed under $Z$. We briefly describe the default ``non-tape'' configuration and the transition function of $\mathfrak{M}^*$. 

The symbols in $\sigma(\mathfrak{M})$ have default interpretations following $\mathfrak{M}$. For all $\ulcorner X \urcorner \in \sigma(\mathfrak{M})$, the default interpretation of $\ulcorner R^{\ulcorner X \urcorner} \urcorner$ is the empty set. The two extra constant symbols $\ulcorner c_0 \urcorner$ and $\ulcorner c_1 \urcorner$ default to $0$ and $1$ respectively. 

At every step of a computation of $\mathfrak{M}^*$, it checks if ``$c_0 = c_1$''. If the answer is no, it copies the previous $\mathfrak{M}$ portion data to the $R^{\ulcorner X \urcorner}$s within the fibers denoted by $c_0$ (which is supposed to count the steps in every short terminating run), simulates the next step of its $\mathfrak{M}$ portion, and increments both $c_0$ and $c_1$. Otherwise, the moral implication is that a limit stage has been reached, whence it first verifies if $Y$ is functional on the current $R^{\ulcorner X \urcorner}$s (note that the $R^{\ulcorner X \urcorner}$s might contain gibberish in states not reachable from any initial state). If so, it computes point-wise limits of the $R^{\ulcorner X \urcorner}$s according to $Y$, copies the resulting limit state information to its $\mathfrak{M}$ portion, and increments just $c_1$. Otherwise, it terminates. We can argue by induction that $\mathfrak{M}^*$ always correctly computes the states of a short terminating run of $\mathfrak{M}$, primarily because the preservation of ordinal limits by the simple limit formulas in $Z$ allows both
\begin{itemize}
    \item the preservation of data previously stored in the $R^{\ulcorner X \urcorner}$s, and
    \item the pair of constants $c_0$ and $c_1$ to flag limit stages of a terminating run.
\end{itemize} 

In conclusion, the $\mathfrak{M}$ portion of $\mathfrak{M}^*$ simulates (with some latency) the entirety of each short terminating run of $\mathfrak{M}$. As said portion contains both the input and output tapes of $\mathfrak{M}^*$, \ref{2393} holds.
\end{proof}

The following definition is now well-motivated.

\begin{defi}
Let $\mathfrak{M}$ and $Y_{\mathrm{lf}} := \{\phi_{lim}^{\ulcorner X \urcorner} : \ulcorner X \urcorner \in \sigma(\mathfrak{M})\}$ be as in Proposition \ref{prop233}. A \emph{short terminating run of} $\mathfrak{M}$ \emph{with input} $A$ \emph{and output} $B$ is a short terminating run of $(\mathfrak{M}, Y_{\mathrm{lf}})$ with input $A$ and output $B$.
\end{defi}

Basically, we are mandating the taking of limit inferiors point-wise wherever possible to be the standard method by which we compute limit states of transfinite runs. This is not a real restriction, for we have established from Propositions \ref{prop233} and \ref{prop239} that our new standard is among the most powerful ways to take limits of a GSeqA's states, while still maintaining absoluteness and the spirit of local definability.

Until now, we have only considered ``short'' transfinite runs, that is, runs that take less steps than the length of the input/output tape. What about longer runs? After all, the ability for Turing machines to run forever is the reason they can ``output'' non-computable computably enumerable sets. Allowing a Turing machine to run transfinitely (e.g. in \cite{ittm}) allows it to explicitly decide these sets. Through such means, oracle computation becomes more straightforward: the oracle can be incorporated into the input, and the result of the machine being fed that oracle can be directly written on the output tape. We want to have an analogous concept of ``long'' runs on GSeqAs. It turns out that removing the length constraint on terminating runs works just fine; 

\begin{defi}\label{def242}
Let $\mathfrak{M}$ be a GSeqA. A \emph{terminating run of} $\mathfrak{M}$ \emph{with input} $A$ \emph{and output} $B$ is a function $Sq$ from some successor ordinal $\delta$ into $S(\mathfrak{M})$ such that 
\begin{enumerate}[label=(\alph*)]
    \item $A, B \subset \kappa(\mathfrak{M})$,
    \item $Sq(0) = l_{\mathfrak{M}}(A)$,
    \item $Sq(\alpha + 1) = \tau_{\mathfrak{M}}(Sq(\alpha)) \neq Sq(\alpha)$ for all $\alpha < \delta - 1$,
    \item for every limit $\gamma < \delta$, $n$-ary function symbol $\ulcorner X \urcorner \in \sigma(\mathfrak{M})$, and $\alpha_1, \dots, \alpha_n \in \kappa(\mathfrak{M})$,
    \begin{gather*}
        \ulcorner X \urcorner^{Sq(\gamma)} (\alpha_1, ..., \alpha_n) \! = \! 
        \begin{cases}
            lim \ inf \{(\beta, \ulcorner X \urcorner^{Sq(\beta)}(\alpha_1, ..., \alpha_n)) \! : \! \beta < \gamma\} & \!\!\!\! \text{if it exists} \\
            0 & \!\!\!\! \text{otherwise,}
        \end{cases} 
    \end{gather*}
    \item $\tau_{\mathfrak{M}}(Sq(\delta - 1)) = Sq(\delta - 1)$, and
    \item $u_{\mathfrak{M}}(Sq(\delta - 1)) = B$.
\end{enumerate}
\end{defi}

A key realisation informing Definition \ref{def242} is that the existence, the uniqueness and the computation of limit inferiors, with regards to sequences of ordinals indexed by ordinals, are absolute. This, coupled with Proposition \ref{prop226}, gives us the absoluteness of terminating runs. More precisely, let $\kappa$ be a limit ordinal and $A, B$ be two sets of ordinals. Due to the absoluteness of taking limit inferior in general, the sets
\begin{align*}
    \mathrm{SRuns}_{\kappa}(A, B) := \{(\mathfrak{M}, Sq) : \ & \mathfrak{M} \text{ is a GSeqA with } \kappa(\mathfrak{M}) = \kappa \text{ and } Sq \text{ is a short} \\
    & \text{terminating run of } \mathfrak{M} \text{ with input } A \text{ and output } B\}
\end{align*}
and
\begin{align*}
    \mathrm{Runs}_{\kappa}(A, B) := \{(\mathfrak{M}, Sq) : \ & \mathfrak{M} \text{ is a GSeqA with } \kappa(\mathfrak{M}) = \kappa \text{ and } Sq \text{ is a} \\
    & \text{terminating run of } \mathfrak{M} \text{ with input } A \text{ and output } B\}
\end{align*}
as defined are absolute for transitive models of $\mathsf{ZFC}$ containing $A, B$ and $\kappa$. 

\begin{rem}\label{rem243}
For a limit ordinal $\kappa$ and two sets of ordinals $A$ and $B$, that $\mathrm{Runs}_{\kappa}(A, B)$ is a set and not a proper class stems from the fact that by our definition of terminating runs, $Sq$ must be an injection into $S(\mathfrak{M})$ whenever $(\mathfrak{M}, Sq) \in \mathrm{Runs}_{\kappa}(A, B)$. The aforementioned fact also means that all terminating runs of a GSeqA $\mathfrak{M}$ must be of length $< (2^{\kappa(\mathfrak{M})})^+$.
\end{rem}

\section{Computability and Comparisons}

In this section, we derive from our generalised sequential algorithms, various notions of generalised (relative) computability. We then compare these notions with existing ones.

\subsection{Relative Computability}

As prefaced in the previous subsection, we aim to use our definition of terminating runs to capture the notion of relative computability right off the bat.

\begin{defi}
For any limit ordinal $\kappa$ and any two sets of ordinals $A$ and $B$, $B$ is $\kappa$\emph{-computable from} $A$ --- denoted $B \leq^C_{\kappa} A$ --- iff
\begin{enumerate}[label=(\alph*)]
    \item $A, B \subset \kappa$, and
    \item for some GSeqA $\mathfrak{M}$ with $\kappa(\mathfrak{M}) = \kappa$, there is a terminating run of $\mathfrak{M}$ with input $A$ and output $B$. 
\end{enumerate}
In this case, $\mathfrak{M}$ \emph{witnesses} $B \leq^C_{\kappa} A$.
\end{defi}

\begin{defi}
We say $B$ is \emph{computable from} $A$, or $B \leq^C A$, iff $B \leq^C_{\kappa} A$ for some limit ordinal $\kappa$.
\end{defi}

By Proposition \ref{prop226}, Remark \ref{rem225} and the absoluteness of how the $\mathrm{Runs}_{\kappa}(A, B)$s are defined, $\leq^C_{\kappa}$ is absolute for transitive models of $\mathsf{ZFC}$ containing $\kappa$. Further, $\leq^C$ is absolute for transitive models of $\mathsf{ZFC}$ with the same ordinals. 

Given any candidate for relative computability, we want to approach a definition for outright computability, i.e. computability without qualification. Perhaps the most natural way to do so is to define computable sets as sets that can be computed with the simplest of inputs. 

\begin{defi}
Let $A$ be a set of ordinals.
\begin{enumerate}[label=(\arabic*)]
    \item $A$ is $\kappa$\emph{-computable} iff $A \leq^C_{\kappa} \emptyset$.
    \item $A$ is \emph{computable} iff $A$ is $\kappa$-computable for some limit ordinal $\kappa$.
\end{enumerate}
\end{defi}

Let us look into some basic properties of the relations $\leq^C_{\kappa}$ and $\leq^C$, in an attempt to reflect on the validity of our notion of (relative) computability. Intuitively, any good notion of relative computability should be transitive. 

\begin{prop}\label{prop245}
Let $\kappa$ be a limit ordinal. Suppose $A \leq^C_{\kappa} B$ and $B \leq^C_{\kappa} C$. Then $A \leq^C_{\kappa} C$.
\end{prop}

\begin{proof}
Let $\mathfrak{M}_1$ witness $B \leq^C_{\kappa} C$ and $\mathfrak{M}_2$ witness $A \leq^C_{\kappa} B$. We shall construct a GSeqA $\mathfrak{M}^*$ with $\kappa(\mathfrak{M}^*) = \kappa$ that essentially simulates $\mathfrak{M}_1$ followed by $\mathfrak{M}_2$, so that $\mathfrak{M}^*$ witnesses $A \leq^C_{\kappa} C$. Set $\sigma(\mathfrak{M}^*)$ to be a disjoint union of $\sigma(\mathfrak{M}_1)$ and $\sigma(\mathfrak{M}_2)$, and 
\begin{equation*}
    \phi^d_{\mathfrak{M}^*} := \phi^d_{\mathfrak{M}_1} \wedge \phi^d_{\mathfrak{M}_2} \text{.}
\end{equation*}

At the beginning, the $\mathfrak{M}^*$ simulates $\mathfrak{M}_1$ in its $\mathfrak{M}_1$ portion, updating this portion of each configuration while keeping the ``non-tape'' $\mathfrak{M}_2$ portion undisturbed. In the background, checks are consistently being done. At the start of each step, before an update to its $\mathfrak{M}_1$ portion is carried out, $\mathfrak{M}^*$ checks if said portion will change going from the current state to the next. If so, it continues to simulate $\mathfrak{M}_1$. Otherwise, it copies the contents of the output tape to the input tape, erases all data from the output tape, and starts having its $\mathfrak{M}_2$ portion simulate $\mathfrak{M}_2$. From then on, the ``non-tape'' $\mathfrak{M}_1$ portion of each configuration is left untouched, so that $\mathfrak{M}^*$ terminates exactly when its simulation of $\mathfrak{M}_2$ ends.

Checking that $\mathfrak{M}^*$ fulfils our requirement is straightforward.
\end{proof}

It makes sense that a machine is more powerful when endowed with more space. This might not quite be the case for GSeqAs, but we can get close to this sort of monotonicity with a mild assumption.

\begin{prop}\label{prop246}
Suppose $\kappa, \kappa'$ are limit ordinals, $\kappa < \kappa'$, $A \leq^C_{\kappa} B$, and $\kappa$ is definable over the structure $(\kappa'; \in)$. Then $A \leq^C_{\kappa'} B$.
\end{prop}

\begin{proof}
Let $\mathfrak{M}$ be a GSeqA witnessing $A \leq^C_{\kappa} B$. Define a GSeqA $\mathfrak{M}'$ with $\kappa(\mathfrak{M}') = \kappa'$ and $\sigma(\mathfrak{M}') = \sigma(\mathfrak{M})$. Since $\kappa$ is definable over $(\kappa'; \in)$, we can relativise both $\phi^d_{\mathfrak{M}}$ and $\phi^{\tau}_{\mathfrak{M}}$ to $\kappa'$ so that if $\ulcorner X \urcorner \in \sigma(\mathfrak{M})$ is an $n$-ary function symbol and $\alpha_1, \dots, \alpha_n \in \kappa'$, then 
\begin{enumerate}[label=(\arabic*)]
    \item the default value of $X(\alpha_1, \dots, \alpha_n)$
    \begin{enumerate}[label=(\alph*)]
        \item follows what is given by $\phi^d_{\mathfrak{M}}$ if $\alpha_1, \dots, \alpha_n \in \kappa$, and
        \item is set to $0$ otherwise, and
    \end{enumerate}
    \item the next-stage value of $X(\alpha_1, \dots, \alpha_n)$
    \begin{enumerate}[label=(\alph*)]
        \item follows the result of computation by $\tau_{\mathfrak{M}}$ if $\alpha_1, \dots, \alpha_n \in \kappa$, and
        \item is set to $0$ otherwise.
    \end{enumerate}
\end{enumerate} 
Set $\phi^d_{\mathfrak{M}'}$ and $\phi^{\tau}_{\mathfrak{M}'}$ to be the relativised versions of $\phi^d_{\mathfrak{M}}$ and $\phi^{\tau}_{\mathfrak{M}}$ respectively. Then $\mathfrak{M}'$ witnesses $A \leq^C_{\kappa'} B$.
\end{proof}

\subsection{Further Modifications}\label{ss320}

The definability requirement in Proposition \ref{prop246} hinders us from concluding the transitivity of $\leq^C$. One can read this as a sign that $\leq^C$, and by extension the $\leq^C_{\kappa}$s, is not quite the right notion of relative computability.

Moreover, it can be unsatisfactory to see that
\begin{equation*}
    \{A \subset \kappa : A \leq^C_{\kappa} B\} \text{ is countable}
\end{equation*}
for all limit ordinals $\kappa$ and $B \subset \kappa$ (which follows from Remark \ref{rem225}). An ideal relative computability relation uniformly definable in $\kappa$ should have its predecessor sets grow along with $\kappa$. Towards such a relation, we shall introduce an augmentation to GSeqAs.

\begin{defi}\label{def836}
An \emph{GSeqA with parameters} (henceforth \emph{GSeqAP}) 
\begin{equation*}
    \mathfrak{M} = (\kappa(\mathfrak{M}), \sigma(\mathfrak{M}), \phi^{\tau}_{\mathfrak{M}}, \phi^d_{\mathfrak{M}}, \vartheta_{\mathfrak{M}})
\end{equation*} 
is defined like a GSeqA, except constants in $\sigma(\mathfrak{M})$ are allowed to evaluate to $(\{\alpha\}, 1)$ under $\vartheta_{\mathfrak{M}}$ for any $\alpha \in \kappa(\mathfrak{M})$. In this case, we call $\alpha$ a \emph{parameter} of $\mathfrak{M}$.
\end{defi}

\begin{rem}\label{rem248}
Fix any limit ordinal $\kappa$. There are $|\kappa|$ many GSeqAPs $\mathfrak{M}$ satisfying $\kappa(\mathfrak{M}) = \kappa$. Also, this set of GSeqAP specifications is absolute for transitive models of $\mathsf{ZFC}$ containing $\kappa$ (cf. Remark \ref{rem225}).
\end{rem}

The definitions and properties so far associated with GSeqAs adapt perfectly to GSeqAPs. We will use the adapted definitions in the context of GSeqAPs without further clarifications.

\begin{defi}\label{def38}
For any limit ordinal $\kappa$ and any two sets of ordinals $A$ and $B$, $B$ is $\kappa$\emph{-computable with parameters from} $A$ --- denoted $B \leq^P_{\kappa} A$ --- iff 
\begin{enumerate}[label=(\alph*)]
    \item $A, B \subset \kappa$, and
    \item for some GSeqAP $\mathfrak{M}$ with $\kappa(\mathfrak{M}) = \kappa$, there is a terminating run of $\mathfrak{M}$ with input $A$ and output $B$. 
\end{enumerate}
In this case, $\mathfrak{M}$ \emph{witnesses} $B \leq^P_{\kappa} A$.
\end{defi}

\begin{defi}\label{def839}
We say $B$ is \emph{computable with parameters from} $A$, or $B \leq^P A$, iff $B \leq^P_{\kappa} A$ for some limit ordinal $\kappa$.
\end{defi}

\begin{defi}
Let $A$ be a set of ordinals.
\begin{enumerate}[label=(\arabic*)]
    \item $A$ is $\kappa$\emph{-computable with parameters} iff $A \leq^P_{\kappa} \emptyset$.
    \item $A$ is \emph{computable with parameters} iff $A$ is $\kappa$-computable with parameters for some limit ordinal $\kappa$.
\end{enumerate}
\end{defi}

Let $\kappa$ be a limit ordinal and $A, B$ be two sets of ordinals. Like $\mathrm{SRuns}_{\kappa}(A, B)$ and $\mathrm{Runs}_{\kappa}(A, B)$, the sets
\begin{align*}
    \mathrm{SRunsP}_{\kappa}(A, B) := \{(\mathfrak{M}, Sq) : \ & \mathfrak{M} \text{ is a GSeqAP with } \kappa(\mathfrak{M}) = \kappa \text{ and } Sq \text{ is a short} \\
    & \text{terminating run of } \mathfrak{M} \text{ with input } A \text{ and output } B\}
\end{align*}
and
\begin{align*}
    \mathrm{RunsP}_{\kappa}(A, B) := \{(\mathfrak{M}, Sq) : \ & \mathfrak{M} \text{ is a GSeqAP with } \kappa(\mathfrak{M}) = \kappa \text{ and } Sq \text{ is a} \\
    & \text{terminating run of } \mathfrak{M} \text{ with input } A \text{ and output } B\}
\end{align*}
have absolute definitions for transitive models of $\mathsf{ZFC}$ containing $A, B$ and $\kappa$. 

Adapting Proposition \ref{prop226} to GSeqAPs, in conjunction with Remark \ref{rem248}, the absoluteness of $\kappa^{< \omega}$, and the absoluteness of the definition of the $\mathrm{RunsP}_{\kappa}(A, B)$s, gives us that $\leq^P_{\kappa}$ is absolute for transitive models of $\mathsf{ZFC}$ containing $\kappa$. As a result, like $\leq^C$, $\leq^P$ is absolute for transitive models of $\mathsf{ZFC}$ with the same ordinals. 

By Remark \ref{rem248}, the predecessor sets
\begin{equation*}
    \{A \subset \kappa : A \leq^P_{\kappa} B\}
\end{equation*}
are of size $|\kappa|$ for all limit ordinals $\kappa$ and $B \subset \kappa$. Therefore, in the $\leq^P_{\kappa}$s we get a class of relative computability relations with predecessor sets that grow uniformly in $\kappa$.

A routine check confirms that the parameterised version of Proposition \ref{prop245} holds with the same proof.

\begin{prop}\label{prop251}
Let $\kappa$ be a limit ordinal. Suppose $A \leq^P_{\kappa} B$ and $B \leq^P_{\kappa} C$. Then $A \leq^P_{\kappa} C$.
\end{prop}

If we substitute the parametrised counterparts of relative computability for the non-parametrised ones in Proposition \ref{prop246}, we can do away with the definability hypothesis.

\begin{prop}\label{prop252}
Suppose $\kappa, \kappa'$ are limit ordinals, $\kappa < \kappa'$ and $A \leq^P_{\kappa} B$. Then $A \leq^P_{\kappa'} B$.
\end{prop}

\begin{proof}
We modify the construction of $\mathfrak{M}'$ in the proof of Proposition \ref{prop246}. To ensure $\kappa$ can be referred to in the relativisation process, we add two fresh constant symbols $\ulcorner c \urcorner$ and $\ulcorner d \urcorner$ to $\sigma(\mathfrak{M'})$ and set
\begin{gather*}
    \vartheta_{\mathfrak{M}'}(\ulcorner c \urcorner) := (\{\kappa\}, 1) \\
    \vartheta_{\mathfrak{M}'}(\ulcorner d \urcorner) := (\kappa', 0) \text{.}
\end{gather*}
As there is no way to formulate $\phi^d_{\mathfrak{M}'}$ such that $\ulcorner c \urcorner$ occurs in it, we choose $\phi^d_{\mathfrak{M}'}$ to be any universally valid sentence over $\sigma(\mathfrak{M'})$. For example, setting all return values to $0$ works. Define $\tau_{\mathfrak{M}'}$ as follows. If ``$d = 0$'' in the current state, make use of $c$ to set up the $\sigma(\mathfrak{M})$ portion of the next state according to the relativised version of $\phi^d_{\mathfrak{M}}$, then increment $d$. Otherwise, do nothing to $d$ and modify the $\sigma(\mathfrak{M})$ portion of the next state according to $\phi^{\tau}_{\mathfrak{M}}$ relativised to $c$. The $\mathfrak{M}'$ defined thus witnesses $A \leq^P_{\kappa'} B$.
\end{proof}

\begin{rem}\label{rem256}
The $\mathfrak{M}'$ defined based on $\mathfrak{M}$ in Proposition \ref{prop252} simulates $\mathfrak{M}$ with finite overhead. That is, given any $A, B \subset \kappa$, if there is an $\alpha$-length terminating run of $\mathfrak{M}$ with input $A$ and output $B$, then for some $n < \omega$ there is an $(\alpha + n)$-length terminating run of $\mathfrak{M}'$ with input $A$ and output $B$.
\end{rem}

Notice that the construction in Proposition \ref{prop252} of $\mathfrak{M}'$ from any $\mathfrak{M}$ witnessing $A \leq^P_{\kappa} B$ and any $\kappa'$ can be made uniform. So let $\mathcal{G}$ be a definable class function mapping each such pair $(\mathfrak{M}, \kappa')$ to such a $\mathfrak{M}'$.

The next proposition is immediate from Propositions \ref{prop251} and \ref{prop252}.

\begin{prop}
Suppose $A \leq^P B$ and $B \leq^P C$. Then $A \leq^P C$.
\end{prop} 

Proposition \ref{prop252}, the transitivity of $\leq^P$, and the nice relationship between limit ordinals $\kappa$ and the sizes of $\leq^P_{\kappa}$ predecessor sets, all point to $\leq^P$, $\leq^P_{\kappa}$ being the superior and more valid notion of relative computability. In fact, they also point to a GSeqAP being a better model of generalised algorithm than a GSeqA. For this reason, the rest of this subsection will focus on $\leq^P$, $\leq^P_{\kappa}$ and GSeqAPs.

We are now able to solidify our case for using point-wise limit inferiors, whenever they exist, to compute limit states of terminating runs which are not short. Let us first generalise the concepts of ordinal limits preservation and simple limit states (Definitions \ref{def234c} and \ref{def234}) to GSeqAPs and sequences of arbitrary lengths.

\begin{defi}\label{257c}
Let $\mathfrak{M}$ be a GSeqAP. Given $\ulcorner X \urcorner \in \sigma(\mathfrak{M})$ an $n$-ary function symbol and $\phi$ a \emph{simple limit formula for} $(\ulcorner X \urcorner, \mathfrak{M})$, we say $\phi$ \emph{preserves ordinal limits} iff interpreting $\ulcorner \in \urcorner$ as $\in$ and $\ulcorner R^{\ulcorner X \urcorner} \urcorner$ as $R_{Sq}^{\ulcorner X \urcorner}$, 
\begin{equation*}
    (\lambda; \in, R_{Sq}^{\ulcorner X \urcorner}) \models \phi[x_1 \mapsto \alpha_1, \dots, x_{n+1} \mapsto \alpha_{n+1}]
\end{equation*}
for all functions $Sq$ from a limit ordinal into $S(\mathfrak{M})$, all limit ordinals 
\begin{equation*}
    \lambda \geq max\{\kappa(\mathfrak{M}), dom(Sq)\}
\end{equation*}
and all $(\alpha_1, \dots, \alpha_{n+1}) \in \kappa(\mathfrak{M})^{n+1}$ such that 
\begin{equation*}
    \alpha_{n+1} = lim \{(\beta, \ulcorner X \urcorner^{Sq(\beta)}(\alpha_1, \dots, \alpha_n)) : \beta < dom(Sq)\} \text{.}
\end{equation*}
\end{defi}

\begin{defi}\label{def256}
Let $\mathfrak{M}$ be a GSeqAP and $Sq$ be a function from a limit ordinal into $S(\mathfrak{M})$. A state $s \in S(\mathfrak{M})$ is a \emph{simple limit of} $Sq$ iff for each $n$-ary function symbol $\ulcorner X \urcorner \in \sigma(\mathfrak{M})$ there is a simple limit formula $\phi_{lim}^{\ulcorner X \urcorner}$ for $(\ulcorner X \urcorner, \mathfrak{M})$ such that, interpreting $\ulcorner \in \urcorner$ as $\in$ and $\ulcorner R^{\ulcorner X \urcorner} \urcorner$ as $R_{Sq}^{\ulcorner X \urcorner}$, 
\begin{equation*}
    \ulcorner X \urcorner^s (\alpha_1, \dots, \alpha_n) = \alpha_{n+1} \iff (\lambda; \in, R_{Sq}^{\ulcorner X \urcorner}) \models \phi[x_1 \mapsto \alpha_1, \dots, x_{n+1} \mapsto \alpha_{n+1}]
\end{equation*}
whenever $\lambda$ is a limit ordinal $\geq max\{\kappa(\mathfrak{M}), dom(Sq)\}$ and $(\alpha_1, \dots, \alpha_{n+1}) \in \kappa(\mathfrak{M})^{n+1}$. In this case, we say $\{\phi_{lim}^{\ulcorner X \urcorner} : \ulcorner X \urcorner \in \sigma(\mathfrak{M})\}$ \emph{witnesses} $s$ \emph{is a simple limit of} $Sq$.
\end{defi}

Similarly, we can generalise Definitions \ref{def237} and \ref{def238}.

\begin{defi}\label{def257p}
Let $\mathfrak{M}$ be a GSeqAP. Then $\mathfrak{M}$ is \emph{strongly closed under} a set $Y$ iff for each $n$-ary function symbol $\ulcorner X \urcorner \in \sigma(\mathfrak{M})$ there is a simple limit formula $\phi_{lim}^{\ulcorner X \urcorner}$ for $(\ulcorner X \urcorner, \mathfrak{M})$, such that 
\begin{equation*}
    Y = \{\phi_{lim}^{\ulcorner X \urcorner} : \ulcorner X \urcorner \in \sigma(\mathfrak{M})\}
\end{equation*}
and whenever $Sq$ is a function from some limit ordinal into $S(\mathfrak{M})$, there is $s \in S(\mathfrak{M})$ witnessed by $Y$ to be a simple limit of $Sq$.
\end{defi}

\begin{defi}\label{def258}
Let $\mathfrak{M}$ be a GSeqAP. Then $\mathfrak{M}$ is \emph{practically closed under} a set $Y$ iff for each $n$-ary function symbol $\ulcorner X \urcorner \in \sigma(\mathfrak{M})$ there is a simple limit formula $\phi_{lim}^{\ulcorner X \urcorner}$ for $(\ulcorner X \urcorner, \mathfrak{M})$, such that 
\begin{equation*}
    Y = \{\phi_{lim}^{\ulcorner X \urcorner} : \ulcorner X \urcorner \in \sigma(\mathfrak{M})\}
\end{equation*}
and whenever $Sq$ is a function from some limit ordinal into $S(\mathfrak{M})$ satisfying
\begin{itemize}
    \item $Sq(0) \in I(\mathfrak{M})$,
    \item $Sq(\alpha + 1) = \tau_{\mathfrak{M}}(Sq(\alpha))$ for all $\alpha < \delta$, and
    \item $Y$ witnesses $Sq(\gamma)$ is a simple limit of $Sq \restriction \gamma$ for all limit $\gamma < \delta$,
\end{itemize}
there is $s \in S(\mathfrak{M})$ witnessed by $Y$ to be a simple limit of $Sq$.
\end{defi}

Note that Definitions \ref{def257p} and \ref{def258} are almost word-for-word the same as Definitions \ref{def237} and \ref{def238}. This only changes made are \begin{itemize}
    \item substituting ``GSeqAP'' for ``GSeqA'' and 
    \item removing the requirement for $\delta$ to be $< \kappa(\mathfrak{M})$.
\end{itemize} 

\begin{prop}\label{prop261c}
Proposition \ref{prop233} holds with ``GSeqAP'' substituted for ``GSeqA'' and the requirement for $\delta$ to be $< \kappa(\mathfrak{M})$ removed.
\end{prop}

\begin{proof}
Any formula that works for all GSeqAs $\mathfrak{M}$ with $\ulcorner X \urcorner \in \sigma(\mathfrak{M})$ in the context of proving Proposition \ref{prop233} (such a formula exists by the last sentence of the proof of Proposition \ref{prop233}) also works here.
\end{proof}

\begin{prop}\label{prop261}
Suppose $\mathfrak{M}$ is a GSeqAP practically closed under $Y$. Then for some signature $\sigma'$, as long as $\mathfrak{M}'$ is a GSeqAP satisfying
\begin{itemize}
    \item $\sigma(\mathfrak{M}') = \sigma'$ and
    \item $\kappa(\mathfrak{M}') = \kappa(\mathfrak{M})$,
\end{itemize}
and $Z$ is a set of simple limit formulas preserving ordinal limits with $\mathfrak{M}'$ strongly closed under $Z$, there is a GSeqAP $\mathfrak{M}^*$ such that
\begin{enumerate}[label=(\alph*)]
    \item $\sigma(\mathfrak{M}^*) = \sigma'$,
    \item $\kappa(\mathfrak{M}^*) = (2^{\kappa(\mathfrak{M})})^+$, and
    \item\label{2613} for all $A, B \subset \kappa(\mathfrak{M}^*)$, a terminating run of $(\mathfrak{M}, Y)$ with input $A \restriction \kappa(\mathfrak{M})$ and output $B \restriction \kappa(\mathfrak{M})$ exists iff a terminating run of $(\mathfrak{M}^*, Z)$ with input $A$ and output $B \restriction \kappa(\mathfrak{M})$ exists.
\end{enumerate}
\end{prop}

\begin{proof}
Let us try to follow the proof of Proposition \ref{prop239} with $\mathfrak{M}^{\dagger} := \mathcal{G}(\mathfrak{M}, (2^{\kappa(\mathfrak{M})})^+)$ specified in place of the arbitrary ``$\mathfrak{M}$'' given by said proposition. Note that 
\begin{itemize}
    \item Remark \ref{rem243} remains true with ``$\mathrm{RunsP}_{\kappa}(A, B)$'' substituted for ``$\mathrm{Runs}_{\kappa}(A, B)$'' and ``GSeqAP'' substituted for ``GSeqA'', and 
    \item \ref{2613} holds with $\mathfrak{M}^{\dagger}$ in place of $\mathfrak{M}^*$.
\end{itemize}
Also observe that by having $\mathfrak{M}^{\dagger}$ as our starting point and choosing $\sigma'$ to be $\sigma(\mathfrak{M}^{\dagger})$, the argument from the second paragraph onwards in the proof of Proposition \ref{prop239} applies here if we just assume the hypotheses for $\mathfrak{M}'$ and $Z$ to be as stated in this proposition. 
\end{proof}

Propositions \ref{prop261c} and \ref{prop261} allude to our method of computing limit states being utmost encompassing computability-wise amidst ``locally definable'' competition, regardless of sequence lengths.

Clearly, $\leq^P$ is reflexive, witnessed by the one-step copying of input tape contents onto the output tape. Hence $\leq^P$ is a preorder. We can quotient $\leq^P$ by a canonical equivalence relation to produce a proper-class-sized partial order analogous to the order of Turing degrees.

\begin{defi}
Let $A, B$ be two sets of ordinals. We define two equivalence relations below.
\begin{enumerate}[label=(\arabic*)]
    \item For any limit ordinal $\kappa$, we say $A \equiv^P_{\kappa} B$ iff $A \leq^P_{\kappa} B$ and $B \leq^P_{\kappa} A$.
    \item We say $A \equiv^P B$ iff $A \leq^P B$ and $B \leq^P A$.
\end{enumerate}
\end{defi}

Taking quotient of $\leq^P$ by $\equiv^P$ would result in a partial order. Denote this partial order by $(\mathcal{D}_P, \leq_{\mathcal{D}_P})$. Members of $\mathcal{D}_P$ are called $P$\emph{-degrees}. We would like to study the structure of $P$-degrees under $\leq_{\mathcal{D}_P}$ the way we study the structure of Turing degrees under Turing reducibility (modulo Turing equivalence). But before that, it would serve us well to consult and compare with existing literature on models of generalised computability. 

If classical recursion or computability theory is about the study of subsets of $\omega$ through Turing reducibility and degrees, then one can envision a natural extension of this study aiming at analogous notions of reducibility and degrees for subsets of ordinals larger than $\omega$. Indeed, there have been a number of successful extensions, lumped under the subject of \emph{higher recursion theory}. The theories of alpha recursion (or $\alpha$-recursion) and $E$-recursion are the most developed among them, but both recursion theories in their typical presentations are built upon schemata akin to G\"{o}del's model of real computation (see \cite{takeuti} and \cite{normann} respectively).

Emerging later into the scene is the theory of ordinal computability, started by Koepke in \cite{koepke1}. A comparison with $\alpha$-recursion theory is given in \cite{koepke2}. We will drawn on parts of this comparative study to give an equivalent definition of $\leq^P$.

\subsection{Associations with Constructibility}

Let us begin this subsection with a recapitulation of selected set-theoretic concepts.

\begin{defi}
We use $L[A]$ to denote the \emph{constructible universe relative to a set} $A$, and $L$ to denote $L[\emptyset]$. For each ordinal $\alpha$, let $L_{\alpha}[A]$ respectively denote the $\alpha$-th level $L[A]$, so $L_{\alpha}[\emptyset]$ is just $L_{\alpha}$. The reader may refer to Chapter 13 of \cite{jech} for more details.
\end{defi}

\begin{defi}
Let $\alpha$ be an ordinal and $A$ be a set. Then $\mathbf{\Sigma_1}(L_{\alpha}[A])$ denotes the set of $X \subset \alpha$ with the following property: 
\begin{quote}
    for some $n < \omega$ there are parameters $y_1, \dots, y_n \in L_{\alpha}[A]$ and a $\Sigma_1$ formula $\phi$ with free variables $x_1, \dots, x_{n}, z$ over the signature $\{\ulcorner \in \urcorner, \ulcorner A \urcorner\}$, such that 
    \begin{equation*}
        X = \{\beta < \alpha : (L_{\alpha}[A]; \in, A) \models \phi[x_1 \mapsto y_1, \dots, x_n \mapsto y_n, z \mapsto \beta]\}
    \end{equation*}
    having $\ulcorner \in \urcorner$ interpreted as $\in$ and $\ulcorner A \urcorner$ as $A$.
\end{quote}
In addition, $\mathbf{\Delta_1}(L_{\alpha}[A])$ denotes the set of $X \subset \alpha$ such that $X \in \mathbf{\Sigma_1}(L_{\alpha}[A])$ and $\alpha \setminus X \in \mathbf{\Sigma_1}(L_{\alpha}[A])$.
\end{defi}

\begin{fact}
For every ordinal $\alpha$ and every set $A$, $\mathbf{\Delta_1}(L_{\alpha}[A]) \subset \mathbf{\Sigma_1}(L_{\alpha}[A]) \subset L[A]$.
\end{fact}

Kripke-Platek set theory, or $\mathsf{KP}$, is a weak fragment of $\mathsf{ZF}$ with close connections to $\alpha$-recursion theory.

\begin{defi}
An ordinal $\alpha$ is \emph{admissible} iff $L_{\alpha} \models \mathsf{KP}$.
\end{defi}

An admissible ordinal has reasonably strong closure properties. For example, it is closed under ordinal addition, multiplication and exponentiation.

We require the concept of $\kappa$-computability with short runs for an upcoming lemma.

\begin{defi}\label{def269p}
Let $\kappa$ be a limit ordinal. We say $B$ is $\kappa$\emph{-computable with parameters from} $A$ \emph{with short runs}, or $B \leq^{P, s}_{\kappa} A$, if there is a GSeqAP $\mathfrak{M}$ witnessing $B \leq^P_{\kappa} A$ and a short terminating run of $\mathfrak{M}$ with input $A$ and output $B$. Here, $\mathfrak{M}$ is said to witness $B \leq^{P, s}_{\kappa} A$.
\end{defi}

Unsurprisingly, we have $\leq^{P, s}_{\kappa}$-analogues of Propositions \ref{prop251} and \ref{prop252}.

\begin{prop}\label{prop270n}
Let $\kappa$ be an admissible ordinal. Suppose $A \leq^{P, s}_{\kappa} B$ and $B \leq^{P, s}_{\kappa} C$. Then $A \leq^{P, s}_{\kappa} C$.
\end{prop}

\begin{proof}
As in the case of Proposition \ref{prop251}. the proof here is exactly the same as the proof of Proposition \ref{prop245}. We should additionally note that the termination run of $\mathfrak{M}^*$ --- as in the proof of Proposition \ref{prop245} --- with input $A$ and output $C$ is still short because it combines two short runs and $\kappa$ is closed under ordinal addition.
\end{proof}

\begin{prop}\label{prop271n}
Suppose $\kappa, \kappa'$ are admissible ordinals, $\kappa < \kappa'$ and $A \leq^{P, s}_{\kappa} B$. Then $A \leq^{P, s}_{\kappa'} B$.
\end{prop}

\begin{proof}
Noting Remark \ref{rem256}, the proof here is exactly the same as the proof of Proposition \ref{prop252}.
\end{proof}

There is a natural follow-up to Definition \ref{def269p}.

\begin{defi}\label{def8329}
Let $A, B$ be sets of ordinals. We say $B$ is \emph{computable with parameters from} $A$ \emph{with short runs}, or $B \leq^{P, s} A$, iff $B \leq^{P, s}_{\kappa} A$ for some limit ordinal $\kappa$.
\end{defi}

It so happens that $\leq^P$ and $\leq^{P, s}$ are fundamentally the same.

\begin{prop}\label{prop271}
Let $A, B$ be sets of ordinals. Then $B \leq^P A$ iff $B \leq^{P, s}_{\kappa} A$ for some regular ordinal $\kappa$.
\end{prop}

\begin{proof}
Clearly, $B \leq^P A$ holds if $B \leq^{P, s}_{\kappa} A$ for some regular ordinal $\kappa$. Now assume $B \leq^P A$. Then $B \leq^P_{\kappa'} A$ for some limit ordinal $\kappa'$. Let $\lambda := (2^{|\kappa'|})^+$, so that $\lambda$ is a regular ordinal. By Proposition \ref{prop252} and Remarks \ref{rem243} and \ref{rem256}, $B \leq^{P, s}_{\lambda} A$.
\end{proof}

\begin{prop}
Let $A, B$ be sets of ordinals. Then $B \leq^P A$ iff $B \leq^{P, s} A$.
\end{prop}

\begin{proof}
Again, the ``if`` portion is obvious, so assume $B \leq^P A$. By Proposition \ref{prop271}, $B \leq^{P, s}_{\kappa} A$ for some regular ordinal $\kappa$, which implies $B \leq^{P, s} A$.
\end{proof}

For a limit ordinal $\alpha$, an $\alpha$-machine according to Koepke (see \cite{koepke2}) is essentially a Turing machine with tape length $\alpha$, oracle length up to $\alpha$, and finitely many fixed ordinal parameters $< \alpha$. We can think of an input to an $\alpha$-machine as a code $\langle X, O \rangle$ of a pair of sets, where $X$ represents the non-oracle component of the input and $O$ represents the oracle. At limit steps the limit inferiors of previous head and state positions, along with the previous contents of each cell, are taken. The machine \emph{halts} iff it terminates in $< \alpha$ many steps.

\begin{prop}\label{prop270}
Let $\alpha$ be a limit ordinal. An $\alpha$-machine $M$ can be simulated by a GSeqAP $\mathfrak{M}$ in the sense that for any $A, B \subset \alpha$, $M$ halts on input $A$ with output $B$ iff $\mathfrak{M}$ has a short terminating run with input $A$ and output $B$.
\end{prop}

\begin{proof}
It is not difficult to check that modifying Example \ref{ex228} by substituting $\alpha$ for $\omega$, and adding oracle reads based on how the oracle is coded into the input (through e.g. adding a third tape), works out just fine. The finitely many ordinal parameters of an $\alpha$-machine can be ported wholesale to a GSeqAP as the latter's ordinal parameters; the way a GSeqAP determines the limit states of a computation also perfectly mimicks how limit configurations of an $\alpha$-machine are defined. Moreover, an $\alpha$-machine halting on input $A$ with output $B$ is equivalent to the GSeqAP simulating it having a short terminating run with input $A$ and output $B$.
\end{proof}

\begin{defi}[Koepke, \cite{koepke2}]
Let $\alpha$ be a limit ordinal. A set $B \subset \alpha$ is $\alpha$\emph{-computable in} $A \subset \alpha$ --- denoted $B \preceq_{\alpha} A$ --- iff there is an $\alpha$-machine $M$ for which
\begin{enumerate}[label=(\alph*)]
    \item $M$ halts on input $\langle \{\beta\}, A \rangle$ for all $\beta < \alpha$, and
    \item $B = \{\beta < \alpha : M(\langle \{\beta\}, A \rangle) = 1\}$.
\end{enumerate}
In this case, $M$ \emph{witnesses} $B$ \emph{is} $\alpha$\emph{-computable in} $A$.
\end{defi}

\begin{fact}[Koepke, \cite{koepke2}]\label{fact334}
Let $\alpha$ be an admissible ordinal. Then $B$ is $\alpha$-computable in $A$ iff $B \in \mathbf{\Delta_1}(L_{\alpha}[A])$.
\end{fact}

Recall the following definition referenced in Subsection \ref{subsec22}.

\begin{defi}
Given sets of ordinals $A$ and $B$, we say $B \preceq_A A$ iff there exists an admissible ordinal $\alpha$ such that $B \preceq_{\alpha} A$.
\end{defi}

\begin{rem}\label{rem336}
It follows rather directly from Fact \ref{fact334} that whenever $A$ and $B$ are sets of ordinals, $B \preceq_A A$ iff $B \in L[A]$. In other words, $\preceq_A$ coincides with the relative constructibility relation.
\end{rem}

\begin{lem}\label{lem269}
Let $B \leq^{P, s}_{\kappa} A$ for some regular ordinal $\kappa$. Then $B \in \mathbf{\Delta_1}(L_{\kappa}[A])$, so $B$ is $\alpha$-computable in $A$ by Fact \ref{fact334}.
\end{lem}

\begin{proof}
Let $\mathfrak{M}$ be a GSeqAP with parameter set $K$ witnessing $B \leq^{P, s}_{\kappa} A$, and $Sq$ be a short terminating run of $\mathfrak{M}$ with input $A$ and output $B$. Suppose we want to find out whether some $\beta < \kappa$ is a member of $B$. To do so, we draw up a dependency tree of terms of which interpretations among members of $ran(Sq)$ are needed to decide if $\beta \in B$. This tree is finitely branching by the bounded exploration postulate, and it has number of levels no more than $dom(Sq)$. By the regularity of $\kappa$, it has size $< \kappa$. 

Now, the statement $x \in B$ is semantically equivalent over $V$ to a $\Sigma_1$ formula with free variable $x$ and parameters among $K \cup \{dom(Sq)\}$ saying that such a (unique) tree of terms and interpretations exist. The construction of said tree, due to its simplicity, can be carried out correctly in $L[A]$ through transfinite recursion. Moreover, $L[A]$ recognises the tree has size $< \kappa$, so indeed $B \in \mathbf{\Sigma_1}(L_{\kappa}[A])$.

We can define a GSeqAP $\mathfrak{M}'$ that copies what $\mathfrak{M}$ does on each input until $\mathfrak{M}$ terminates, and flips every bit of the output tape thereafter. Now $\mathfrak{M}'$ witnesses $\kappa \setminus B \leq^{P, s}_{\kappa} A$, so by the argument in the preceding paragraph we have $\kappa \setminus B \in \mathbf{\Sigma_1}(L_{\kappa}[A])$.
\end{proof}

\begin{defi}[Koepke, \cite{koepke2}]
Let $\alpha$ be a limit ordinal. A set $B \subset \alpha$ is $\alpha$\emph{-computably enumerable in} $A \subset \alpha$ iff there is an $\alpha$-machine $M$ for which
\begin{equation*}
    B = \{\beta < \alpha : M \text{ halts on input } \langle \{\beta\}, A \rangle\} \text{.}
\end{equation*}
In this case, $M$ \emph{witnesses} $B$ \emph{is} $\alpha$\emph{-computably enumerable in} $A$.
\end{defi}

\begin{fact}[Koepke, \cite{koepke2}]\label{fact272}
Let $\alpha$ be an admissible ordinal. Then $B$ is $\alpha$-computably enumerable in $A$ iff $B \in \mathbf{\Sigma_1}(L_{\alpha}[A])$.
\end{fact}

\begin{lem}\label{lem273}
Let $B \in \mathbf{\Sigma_1}(L_{\kappa}[A])$ for some admissible ordinal $\kappa$. Then $B \leq^P_{\kappa} A$.
\end{lem}

\begin{proof}
By Fact \ref{fact272}, $B$ is $\kappa$-computably enumerable in $A$, as witnessed by some $\kappa$-machine $M$. Let $\mathfrak{M}$ be a GSeqAP simulating $M$ as per Proposition \ref{prop270}. We define a GSeqAP $\mathfrak{M}'$ that, when given input $A$, dovetails and runs $M$ through $\mathfrak{M}$ on each input in
\begin{equation*}
    \{\langle \{\beta\}, A \rangle : \beta < \kappa\}
\end{equation*}
for arbitrarily many steps $< \kappa$, in order to determine if $M$ halts on said input. In the end, $\mathfrak{M}'$ should consolidate the findings and output them as a subset of $\kappa$.

More formally, add to $\sigma(\mathfrak{M})$ fresh constant symbols $\ulcorner c_0 \urcorner$, $\ulcorner c_1 \urcorner$, $\ulcorner c_2 \urcorner$, $\ulcorner d \urcorner$ and a fresh unary function symbol $\ulcorner R \urcorner$, and have the result be $\sigma(\mathfrak{M}')$. Set all constants to $0$ and $R$ to be empty by default. We are left to describe the transition function of $\mathfrak{M}'$, which we shall do by cases. In our description, $I$ denotes the input to $\mathfrak{M}'$.
\begin{enumerate}[leftmargin=60pt, label=Case \arabic*:]
    \item ``$c_0 = d = 0$''. We increment both $c_0$ and $d$ and set the input tape to display $\langle \emptyset, I \rangle$.
    \item ``$c_0 \neq 0 \wedge d \neq 0$''. We do the dovetailing. In other words, for each $\beta$ satisfying ``$\neg R(\beta) \wedge \beta \leq c_0$'', we use $\mathfrak{M}$ to simulate $M$ on input $\langle \{\beta\}, I \rangle$ for up to $c_0$ steps, using $c_1$ and $c_2$ to keep track of the counts. At the end of every step in simulation, we check if $\mathfrak{M}$ terminates. If so, we set $R(\beta)$ to \texttt{true} and move on to simulate $M$ on input $\langle \{\beta + 1\}, I \rangle$, if $\beta \neq c_0$. At the end of all the simulations and possible updates to $R$, set $c_1$ and $c_2$ to $0$ and increment $c_0$. The constant $d$ is left unchanged.
    \item\label{8340c3} ``$c_0 = 0 \wedge d \neq 0$''. We copy the contents of $R$ to the output tape, before terminating.
    \item\label{8340c4} ``$c_0 \neq 0 \wedge d = 0$''. Just terminate.
\end{enumerate}
Every terminating run of $\mathfrak{M}'$ is of length $\kappa + 2$, and if its input is $I$, then its output will be the set
\begin{equation*}
    \{\beta < \alpha : M \text{ halts on input } \langle \{\beta\}, I \rangle\} \text{.}
\end{equation*}This is because
\begin{enumerate}[label=(\arabic*)]
    \item no state where \hyperref[8340c4]{Case 4} applies is reachable from an initial state, so a run can only terminate from an initial state as a consequence of \hyperref[8340c3]{Case 3}, and
    \item $c_0$ first returns to $0$ in a state other than an initial state after exactly $\kappa + 1$ many steps, at which point $M$ would have been simulated on every input in 
    \begin{equation*}
        \{\langle \{\beta\}, I \rangle : \beta < \kappa\}
    \end{equation*}
    for up to arbitrarily many steps $< \kappa$, with all eventual haltings recorded in $R$.
\end{enumerate}
In particular, there is a terminating run of $\mathfrak{M}'$ with input $A$ and output $B$, whence $B \leq^P A$.
\end{proof}

Lemmas \ref{lem269} and \ref{lem273} reflect the difference in computability between considering all terminating runs of GSeqAPs and considering only short terminating runs, for regular ordinal tape lengths. In tandem, the lemmas also give us a nice characterisation of the reducibility relation $\leq^P$.

\begin{thm}\label{thm275}
Let $A$ and $B$ be two sets of ordinals. Then $B \leq^P A$ iff $B \in L[A]$.
\end{thm}

\begin{proof}
Assume $B \in L[A]$. Then $B \in L_{\kappa}[A]$ for some admissible ordinal $\kappa$. But $L_{\kappa}[A] \subset \mathbf{\Sigma_1}(L_{\kappa}[A])$, so by Lemma \ref{lem273}, $B \leq^P A$. 

Next, assume $B \leq^P A$. Then by Proposition \ref{prop271}, $B \leq^{P, s}_{\kappa} A$ for some regular ordinal $\kappa$. Lemma \ref{lem269} tells us $B \in \mathbf{\Delta_1}(L_{\kappa}[A])$, but $\mathbf{\Delta_1}(L_{\kappa}[A]) \subset L[A]$, so we are done.
\end{proof}

\begin{cor}
A set of ordinals $B$ is computable with parameters iff $B \in L$.
\end{cor}

\begin{proof}
Apply Theorem \ref{thm275} with $\emptyset$ in place of $A$. 
\end{proof}

\begin{cor}
Let $A$ and $B$ be two sets of ordinals. Then $B \leq^P A$ iff $B \preceq_A A$.
\end{cor}

\begin{proof}
Immediate by Remark \ref{rem336}.
\end{proof}

\section{Discussion and Questions}

The following are two basic facts about relative constructibility.

\begin{fact}\label{fact276}
Let $A, B$ be sets of ordinals. Then $B \in L[A]$ iff $L[B] \subset L[A]$.
\end{fact}

\begin{fact}\label{fact277}
Let $X$ be any set. Then there exists a set of ordinals $A$ such that $L[X] = L[A]$.
\end{fact}

From a second-order viewpoint of $V$, we can consider the multiverse of constructible worlds relative to sets in $V$, i.e.
\begin{equation*}
    \mathbf{M}_C := \{L[X] : X \in V\} \text{.}
\end{equation*}
By Theorem \ref{thm275} and Facts \ref{fact276} and \ref{fact277}, 
\begin{equation*}
    (\mathcal{D}_P, \leq_{\mathcal{D}_P}) \cong (\mathbf{M}_C, \subset) \text{.}
\end{equation*}
Philosophically, this seems to indicate that the amount of computing power a set holds is equivalent to how it can extend a particular model of set theory through transfinite recursion. This is far from an original take, for if we take the notion of computability to an extreme, we can argue that any construction via transfinite recursion is a generalised algorithm. What might be slightly more surprising is, by only allowing transfinite recursion that involves very simple formulas ---
\begin{itemize}
    \item either formulas conservatively extending transition functions of classical Turing machines to incorporate limit stage computations,
    \item or formulas satisfying a version of Gurevich's bounded exploration postulate
\end{itemize}
--- we can still achieve the same reach. Indeed, the fact that the relative computability relations arising from
\begin{enumerate}[label=(GCT\arabic*), leftmargin=50pt]
    \item a ``low-level'' machine model of generalised computation such as ordinal Turing machines, 
    \item a ``high-level'' notion of abstract algorithms such as GSeqAPs, and
    \item an algebraic model of set generation such as relativised constructibility
\end{enumerate}
happen to coincide, is suggestive of a plausible generalised-computability counterpart to the Church-Turing thesis.

Just as practical implementations of finitary SeqAs do not aim to model algorithmic intuition beyond the limits of Church-Turing thesis, GSeqAPs should not allow for algorithms that are non-constructible. We can view GSeqAPs as a paradigm for easy high-level descriptions of algorithms on arbitrary sets, one that gives the user enough structure to organise their ideas, and serves as guardrails to rule out clearly non-algorithmic descriptions.

It is consistent that $(\mathcal{D}_P, \leq_{\mathcal{D}_P})$ is highly non-trivial; in fact, this level of non-triviality is implied by suitable large cardinals.

\begin{fact}\label{fact284}
\leavevmode
\begin{enumerate}[label=(\arabic*)]
    \item 
    \!
    $\begin{aligned}[t]
        \mathsf{ZFC} + \text{``} \exists \text{ a transitive model of } \mathsf{ZFC} \text{''} \vdash \text{``} & \exists \text{ a transitive model of } \mathsf{ZFC} \ + \\
        & \text{`} V \neq L[A] \text{ for all } A \in V \text{'''.}
    \end{aligned}$
    \item\label{432} $\mathsf{ZFC} + \text{``} \exists \text{ a strongly compact cardinal''} \vdash \text{``} V \neq L[A] \text{ for all } A \in V \text{''}$.
\end{enumerate}
\end{fact}

In particular, \ref{432} of Fact \ref{fact284} is implied by Exercise 20.2 of \cite{jech}. Note that $V \neq L[A]$ for all $A \in V$ is equivalent to $(\mathcal{D}_P, \leq_{\mathcal{D}_P})$ having no maximal element, so Fact \ref{fact284} actually gives us

\begin{enumerate}[label=(E\arabic*)]
    \item\label{e1}
    \!
    $\begin{aligned}[t]
        \mathsf{ZFC} + \text{``} \exists \text{ a transitive model of } \mathsf{ZFC} \text{''} \vdash \text{``} & \exists \text{ a transitive model of } \mathsf{ZFC} \ + \\
        & \text{`} (\mathcal{D}_P, \leq_{\mathcal{D}_P}) \text{ has no maximal element''',}
    \end{aligned}$
    \item $\mathsf{ZFC} + \text{``} \exists \text{ a strongly compact cardinal''} \vdash \text{``} (\mathcal{D}_P, \leq_{\mathcal{D}_P}) \text{ has no maximal element''}$.
\end{enumerate}
If $(\mathcal{D}_P, \leq_{\mathcal{D}_P})$ has no maximal element, then by applications of iterated forcing over $L$, one can uncover a rich lattice structure of $(\mathcal{D}_P, \leq_{\mathcal{D}_P})$.

Let us now turn our attention to the bounded relations $\leq^{P, s}_{\alpha}$ and $\leq^P_{\alpha}$. Let $\leq_{\alpha}$, $\leq_{w \alpha}$ and $\preceq_{\alpha}$ respectively denote the relations ``$\alpha$-recursive in'', ``weakly $\alpha$-recursive in'' and ``$\alpha$-computable in''. The first two relations are well-studied in the once-bustling field of $\alpha$-recursion theory (see e.g. Part C of \cite{sacks}), whereas the last relation comes from the emerging field of $\alpha$-computability theory. It is known (Section 3 of \cite{koepke2}) that 
\begin{equation*}
    \leq_{\alpha} \ \subset \ \leq_{w \alpha} \ \subset \ \preceq_{\alpha}
\end{equation*}
for any admissible ordinal $\alpha$, and the inclusions cannot be reversed in general. By Fact \ref{fact272} and Lemma \ref{lem273}, we have
\begin{equation}\label{eq2}
    \leq_{\alpha} \ \subset \ \leq_{w \alpha} \ \subset \ \preceq_{\alpha} \ \subset \ \leq^P_{\alpha}
\end{equation}
for admissible ordinals $\alpha$. The last inclusion cannot be reversed in general because $\leq^P_{\alpha}$ is transitive for all admissible $\alpha$ but $\preceq_{\alpha}$ is not. On the other hand, it is clear that $\leq^{P, s}_{\alpha} \ \subset \ \leq^P_{\alpha}$ for all limit $\alpha$, and the reverse is not true due to Lemmas \ref{lem269} and \ref{lem273}. Further, for admissible $\alpha$, $\leq_{w \alpha} \ \neq \ \leq^{P, s}_{\alpha} \ \neq \ \preceq_{\alpha}$ in general because $\leq^{P, s}_{\alpha}$ is transitive unlike the other two relations. We know $\leq^{P, s}_{\alpha} \ \neq \ \leq_{\alpha}$ in general because $\leq^{P, s}_{\alpha}$ is upward consistent in the sense of Proposition \ref{prop252}, but $\leq_{\alpha}$ is not by Theodore Slaman's answer in \cite{slaman}. 

Table \ref{table3} summarises the similarities and differences among the relations $\leq^{P, s}_{\alpha}$, $\leq^P_{\alpha}$, $\leq_{\alpha}$ and $\preceq_{\alpha}$. When we say a relation possesses the property of ``oracle-analogue'', we mean that the definition of said relation depends on --- or makes reference to --- some object analogous to the oracle in the definition of classical oracle machines.

\begin{table}[!ht]
    \caption[Comparing our new relative computability relations against more standard relations]{Comparing two of our new relative computability relations against more standard relations in the study of generalised computability, at an arbitrary admissible ordinal $\alpha$.}
    \label{table3}
    \centering
    \begin{tabular}{|l||*{4}{c|}}\hline
        \backslashbox[110pt]{\footnotesize Property}{\footnotesize Relation}
        &\makebox[4em]{$\leq_{\alpha}$}&\makebox[4em]{$\preceq_{\alpha}$}
        &\makebox[2em]{$\leq^P_{\alpha}$}&\makebox[2em]{$\leq^{P, s}_{\alpha}$}\\\hline\hline
        Oracle-analogue? & \ding{51} & \ding{51} & \ding{55} & \ding{55} \\\hline
        Transitive? & \ding{51} & \ding{55} & \ding{51} & \ding{51} \\\hline
        Upward consistent? & \ding{55} & \ding{51} & \ding{51} & \ding{51} \\\hline
        Appears in$\dots$ & $\alpha$-recursion & $\alpha$-computability & Def. \ref{def38} & Def. \ref{def269p} \\\hline
    \end{tabular}
\end{table}

\begin{ques}
Is the set 
\begin{equation*}
    \{\leq_{\alpha}, \leq_{w \alpha}, \preceq_{\alpha}, \leq^{P, s}_{\alpha}\}
\end{equation*}
linearly ordered by inclusion, for all admissible $\alpha$?
\end{ques}

\section{References}
\printbibliography[heading=none]

\end{document}